\documentclass[journal,draftcls,onecolumn,12pt,twoside]{IEEEtran}

\usepackage[pdftex]{graphicx} 
\usepackage{bbm}

\usepackage{mathrsfs}  
\usepackage{amsmath, bm}

\usepackage{microtype}
\usepackage{graphicx}
\usepackage{authblk}
\usepackage[indent=10pt]{parskip}
\usepackage{enumerate}
\usepackage[normalem]{ulem}
\usepackage{hyperref}
\usepackage{upgreek}
\usepackage{amssymb}
\usepackage{mathtools}
\usepackage{amsthm}
\usepackage{xcolor}
\usepackage[thinc]{esdiff}
\usepackage{tikz}
\usetikzlibrary{arrows,calc,shapes,decorations.pathreplacing,patterns,chains,shapes.multipart,arrows.meta,positioning}
\usepackage{xcolor}
\usepackage{caption}
\usepackage{subcaption}
\usepackage{pgfplots}
\usepgfplotslibrary{fillbetween}
\usepackage{float}
\usepackage{algorithm}
\usepackage{algorithmic}

\newcommand{\E}{\mathbb{E}}
\newcommand{\Var}{\operatorname{Var}}
\newcommand{\dd}{\mathrm{d}}

\newcommand{\darrow}{\:\to_{\rm d}\:}
\newcommand{\step}[1]{\langle #1\rangle}
\newcommand{\bg}{\mathrm{bg}}

\newtheorem{theorem}{Theorem}
\newtheorem{corollary}[theorem]{Corollary}
\newtheorem{lemma}[theorem]{Lemma}
\newtheorem{proposition}[theorem]{Proposition}

\newtheorem{remark}{Remark}

\newtheorem{assumption}{Assumption}

\allowdisplaybreaks

\title{{Structure-Aware Online Optimization of the Age of Information with Unknown System Statistics}}
\author{Yoshiaki Inoue, Michel Mandjes, and Liron Ravner
\thanks{This work was supported in part by JSPS KAKENHI Grant Number
24K14839, Israel Science Foundation (ISF) Grant Number\ 1361/23, and by the NWO Gravitation project {\sc networks} under grant agreement no.\ 024.002.003.}
\thanks{
Y.\ Inoue is with 
Department of Information and Communications Technology, 
Graduate School of Engineering, Osaka University, Suita 565-0871, Japan 
(e-mail: yoshiaki@comm.eng.osaka-u.ac.jp).

M.\ Mandjes is with Mathematical Institute, Leiden University, P.O. Box 9512, 2300 RA Leiden, 
The Netherlands (e-mail: m.r.h.mandjes@math.leidenuniv.nl). He is also affiliated with Korteweg-de Vries Institute for Mathematics, University of Amsterdam, Amsterdam, The Netherlands; Eurandom, Eindhoven University of Technology, Eindhoven, The Netherlands; Amsterdam Business School, Faculty of Economics and Business, University of Amsterdam, Amsterdam, The Netherlands.

L.\ Ravner is with the Department of Statistics, University of Haifa, Israel (e-mail: lravner@haifa.stat.ac.il).
}}
\date{\today}

\begin{document}

\maketitle

\begin{abstract}
   Timely information dissemination is essential in modern communication systems, where low informational latency is critical for networked control, monitoring, and real-time decision-making. The Age of Information (AoI) has emerged as a fundamental metric for quantifying information freshness at the receiver. A key challenge in AoI optimization is selecting the update generation rate: infrequent updates lead to stale information, whereas overly frequent updates increase queueing delays. We develop a structure-aware, data-driven framework for adaptive AoI optimization modeled as single-server queues with unknown system statistics. Our key idea is to exploit analytical representations of the AoI gradient, which depend on a finite-dimensional collection of system quantities that can be estimated from observations collected during system operation. This reduces the online learning problem to estimating the quantities required by the gradient, rather than learning the AoI objective as a black box. We propose an online stochastic approximation algorithm that jointly estimates these quantities and optimizes the update generation rate. We prove almost sure convergence of the algorithm to the optimal update generation rate under mild regularity and moment conditions, despite the temporal dependence induced by the underlying queueing process. We demonstrate the framework in four representative incomplete-information settings, including FCFS queues with unknown service characteristics or hidden background traffic, a preemptive LCFS system with partial observations, and a generate-at-will system with threshold waiting.
\end{abstract}

\section{Introduction}

Modern communication systems are evolving from platforms for reliable data delivery to infrastructures for continuous information sharing. Rather than transferring a file or message once, many emerging applications require a remote monitor to maintain an accurate and up-to-date view of a dynamic physical process. Representative examples include industrial automation, intelligent transportation systems, environmental monitoring, and Internet-of-Things (IoT) applications. In these applications, decisions are made based on the most recently available information, making information freshness a primary system objective.

The \emph{Age of Information} (AoI), introduced as a measure of information freshness \cite{Kaul2012,Sun2017,Yates2019,YatesSun2021}, has emerged as a fundamental performance metric for quantifying timeliness. At any given time, the AoI represents the elapsed time since the generation of the most recently received update. Unlike classical metrics such as throughput or packet delay, AoI captures the combined effects of update generation, queueing, transmission, and service. Over the past decade, AoI has motivated an extensive literature on scheduling policies \cite{Bedewy2018LGFS,Kaul2012}, queueing systems \cite{Inoue2019,Inoue2025,Tripathi2019}, resource allocation \cite{NinoMora2023,Sun2019Sampling}, and energy-constrained communication networks \cite{Qin2024EnergyAoI}.

A fundamental design problem in AoI-aware communication systems is to determine how frequently status updates should be generated. Generating updates too infrequently leads to stale information, whereas generating them too aggressively increases congestion and queueing delays. Consequently, the average AoI is typically minimized at an intermediate operating point that balances these opposing effects \cite{Kaul2012,Sun2016UpdateOrWait,Yates2019}. Analytical work on AoI characterizes this optimum for a variety of queueing systems under the assumption that the underlying system parameters, such as service-time distributions, channel characteristics, and traffic statistics, are known \emph{a priori} \cite{Inoue2019}. In practice, however, these quantities are often unknown or slowly varying over time, making purely model-based optimization difficult.

This limitation has motivated a growing literature on AoI optimization and sampling under unknown system statistics. Tang et al.~\cite{TangChenWangYangTassiulas2023} consider age-optimal sampling when the delay distribution is unknown and develop an online Robbins--Monro procedure that learns the optimal sampling policy directly from observed transmission delays. Related work considers remote estimation and sampling when channel-delay statistics are unknown, again exploiting stochastic approximation to adapt the sampling policy from observations \cite{ChenWangTangTassiulas2023}. More generally, learning-based approaches have been developed for AoI minimization when channel or system dynamics are unknown, including reinforcement-learning formulations that learn age-optimal policies from operational data \cite{Sun2017,YatesSun2021}. These studies demonstrate that online learning can remove the need for a priori knowledge of system statistics. Our approach differs in that it explicitly exploits an analytical representation of the AoI objective: rather than learning the objective or optimal policy as a black box, we estimate the finite-dimensional quantities appearing in its gradient and optimize using this structural representation.

We build on this perspective by exploiting the analytical structure of the AoI objective. Our key observation is that, for a broad class of queueing systems, queueing analysis provides an explicit representation of the gradient of the average AoI with respect to the update generation rate. Crucially, the gradient depends only on a finite-dimensional collection of system quantities, each of which can be estimated from observations collected during system operation. Thus, rather than learning the entire AoI objective or treating the system as a black box, we reduce the statistical learning problem to estimating the specific quantities required by the analytical gradient and directly optimize using these estimates.

We consider AoI optimization over communication channels modeled as single-server queues with unknown system parameters. We propose an online stochastic approximation framework that jointly estimates the unknown quantities appearing in the AoI gradient and optimizes the update generation rate. By combining the analytical gradient representation with online estimation, the algorithm operates directly on observations available at the receiver and requires no prior knowledge of the unknown system parameters. This structure-aware approach avoids treating the queueing system as a black box while remaining applicable when its statistical characteristics are initially unknown or slowly varying over time \cite{Borkar2008,Bottou2018,RobbinsMonro1951,KY2003}.

The proposed framework presents a number of technical challenges. Unlike classical stochastic gradient methods, the observations used for parameter estimation are generated by an evolving queue and are therefore temporally dependent through the underlying queueing dynamics. Consequently, standard stochastic approximation arguments based on independent or exogenous observations do not directly apply. Our analysis explicitly accounts for this endogenous temporal dependence and establishes rigorous convergence guarantees for the resulting online optimization algorithm.

The main contributions of this paper are summarized as follows:

\begin{itemize}

\item We formulate a general online AoI optimization problem in which the update generation rate affects both packet generation and queueing dynamics, while the relevant system statistics are unknown.

\item We propose a structure-aware stochastic approximation algorithm that exploits analytical representations of the AoI gradient. Unlike black-box learning approaches, the proposed method estimates only a finite-dimensional collection of unknown quantities appearing in the gradient formula, rather than learning the entire system dynamics.

\item We establish almost sure convergence of the proposed algorithm to the optimal update generation rate under mild regularity conditions. Unlike standard stochastic approximation settings with exogenous noise, our analysis must account for the intricate endogenous temporal dependence induced by the evolving queueing process, whose state is itself coupled to the update generation rate being optimized.

\item We demonstrate the versatility of the proposed framework through four representative incomplete-information settings, including systems with unknown service-time distributions, hidden background traffic, partial observations, and threshold-based update policies. These examples illustrate that substantially different queueing and observation structures can be treated within a unified optimization framework.

\end{itemize}

The rest of this paper is organized as follows. Section~\ref{sec:model}
introduces the system model, formulates the optimization problem, and
presents representative examples. Section~\ref{sec:alg} develops the
proposed algorithm and establishes its convergence to the optimal update
generation rate. Section~\ref{sec:application} shows that the examples are
covered by the general framework. Section~\ref{sec:num} presents numerical
experiments, and Section~\ref{sec:conclusion} discusses the results and
outlines directions for future work.

\section{Model, problem statement, leading examples}\label{sec:model}

In this section, we first introduce the general information update system and fix the notation used throughout the paper. We then formulate the structure-aware, data-driven AoI optimization problem considered in this paper. The section concludes with four representative examples illustrating how the framework applies to different queueing and observation settings.


\subsection{Information Update System}

We consider an information update system that is composed of a sender (sensor), a
communication channel or processor (server), and a receiver (monitor). The
sensor intermittently generates information packets and transmits them to the
server, which processes each packet and forwards it to the monitor upon
completion of service.

In this general formulation, the key timing variables are the
packet-generation intervals $G_n$ ($n = 1,2,\ldots$) and the system delays
$D_n$ ($n = 0,1,\ldots$). The interval $G_n$ denotes the elapsed time between
the $(n-1)$-st and $n$-th packet generations, whereas $D_n$ represents the
end-to-end delay experienced by the $n$-th packet from its generation to its
reception at the monitor.

Let $(\alpha_n)_{n=0,1,\ldots}$ denote the packet-generation times and
$(\beta_n)_{n=0,1,\ldots}$ the corresponding reception times. These sequences
are defined recursively as
\begin{align*}
\alpha_n &:= \alpha_{n-1} + G_n, \quad n = 1,2,\ldots,\\
\beta_n &:= \alpha_n + D_n, \quad n = 0,1,\ldots,
\end{align*}
with a given initial generation time $\alpha_0$.

In general, not all generated packets contribute new information to the
monitor. In particular, if a packet is overtaken by a more recently generated
packet, it does not trigger a monitor update. We therefore define $\mathcal{I}$
as the set of {\it informative} packets:
\[
\mathcal{I}
:= \{n \in \{0,1,\ldots\};\, \beta_n < \min\{\beta_{n+1}, \beta_{n+2}, \ldots\}\}.
\]
We introduce the subsequences
$(\alpha_{\ell}^{\dag})_{\ell=0,1,\ldots}$ and
$(\beta_{\ell}^{\dag})_{\ell=0,1,\ldots}$ of
$(\alpha_n)_{n=0,1,\ldots}$ and $(\beta_n)_{n=0,1,\ldots}$,
respectively, corresponding to the informative packets:
\begin{align*}
\alpha_{\ell}^{\dag} &:= \alpha_{n(\ell)}, 
\quad
\beta_{\ell}^{\dag} := \beta_{n(\ell)}, \qquad \ell = 0,1,\ldots,
\end{align*}
where $n(\ell)$ denotes the $\ell$-th smallest element of $\mathcal{I}$.
We define $G_{\ell}^{\dag}$ and $D_{\ell}^{\dag}$ as the
inter-generation time and system delay of informative packets:
\begin{align*}
G_{\ell}^{\dag} &:= \alpha_{\ell}^{\dag} - \alpha_{\ell-1}^{\dag},
\quad
\ell = 1,2,\ldots,
\\
D_{\ell}^{\dag} &:= 
\beta_{\ell}^{\dag} - \alpha_{\ell}^{\dag} 
,
\quad
\ell = 0,1,\ldots.
\end{align*}

We assume $\beta_0^{\dag} \leq 0$ without loss of generality.
Let $M_t$ denote the index of the latest informative packet received before time
$t$:
\begin{equation}
M_t := \max\left\{\ell \in \{0,1,\ldots\};\, \beta_{\ell}^{\dag} \leq t \right\},
\quad
t \geq 0.
\label{eq:M_t-def}
\end{equation}
The AoI at time $t$ is then given by
\begin{align*}
A_t &:= t - \alpha_{M_t}^{\dag}
= t - \max\left\{\alpha_{\ell}^{\dag},\, \ell \in \{0,1,\ldots\};\, \beta_{\ell}^{\dag} \leq t\right\},
\quad
t \geq 0.
\end{align*}
Notice that the assumption $\beta_0^{\dag} \leq 0$ above is necessary for the
AoI to be well-defined for each $t \geq 0$.

\begin{remark}
The above construction presents the AoI process in a general form through the
sequences of packet-generation intervals and system delays,
$(G_n,D_n)_{n=0,1,\ldots}$. In specific queueing models, however, a system is
typically specified through its stochastic primitives, such as interarrival
times, service times, and the service discipline, which induce the
sequences of generation and reception times
$(\alpha_n,\beta_n)_{n=0,1,\ldots}$. Several concrete examples of such queueing
models are presented in Section~\ref{ssec:examples}.
\end{remark}

For the time-average of the AoI, it is known \cite{Kaul2012} that, almost surely (a.s.),
\[
\frac{1}{\tau}\int_0^{\tau} A_t \,\dd t
=
\frac{M_{\tau}}{\tau}
\cdot
\frac{1}{M_{\tau}}
\sum_{\ell=1}^{M_{\tau}}
\left(
G_{\ell}^{\dag} D_{\ell}^{\dag} + \frac{(G_{\ell}^{\dag})^2}{2}
\right)
+
O(1/\tau),
\]
as $\tau \to \infty$.
In particular, under suitable regularity conditions, we have
\begin{equation}
\E[A] = 
\frac{\E[G_{\ell}^{\dag} D_{\ell}^{\dag}]}{\E[G^{\dag}]} + \frac{\E[(G^{\dag})^2]}{2\E[G^{\dag}]},
\label{eq:EA-by-GD}
\end{equation}
where $A$ and $G^{\dag}$ denote generic random variables for $A_t$ and
$G_{\ell}^{\dag}$.
Note that we have included the subscript $\ell$ in $\E[G_{\ell}^{\dag} D_{\ell}^{\dag}]$ 
to indicate the dependence between the random objects $G_{\ell}^{\dag}$ and $D_{\ell} ^{\dag}$.

\subsection{Problem Statement}
\label{ssec:problem}

In this paper, we consider a setting in which the generation intervals
$(G_n)_{n=1,2,\ldots}$ are
parametrically controlled, and in which the delay dynamics depend on that control
parameter. Specifically, we make the following assumptions:

\begin{enumerate}
\item \textit{Parametric generation interval.}
The inter-generation times $(G_n)_{n=1,2,\ldots}$ are governed by a parametric family indexed by a scalar parameter $\lambda$. The monitor can adjust $\lambda$ on a time scale slower than the packet-generation and service dynamics, so that $\lambda$ acts as a slowly varying control parameter relative to the underlying queueing process. The parameter $\lambda$ therefore serves as a controllable knob for regulating the generation rate. To make the dependence on $\lambda$ explicit, we rewrite the time-average AoI \eqref{eq:EA-by-GD} as
\begin{equation}\label{eq:EA-by-GD-lambda}
    \E_{\lambda}[A]
    = \frac{\E_{\lambda}\bigl[G_{\ell}^{\dag} D_{\ell}^{\dag}\bigr]}{\E_{\lambda}\bigl[G^{\dag}\bigr]}
    + \frac{\E_{\lambda}\bigl[(G^{\dag})^2\bigr]}{2\,\E_{\lambda}\bigl[G^{\dag}\bigr]}.
\end{equation}

\item \textit{Unknown dependence of delay dynamics on the generation parameter.}
The system delays $(D_n)_{n=0,1,\ldots}$ are influenced by the generation parameter $\lambda$ in a nontrivial manner, i.e., their probability law depends on $\lambda$. However, the exact form of this dependence is not fully known to the controller. Our objective is to choose $\lambda$ so as to minimize the time-average AoI \eqref{eq:EA-by-GD-lambda}, despite having only limited information about how $\lambda$ affects the delay dynamics.

\item \textit{Structural knowledge of the AoI gradient and estimability from informative observations.}
Although the full dependence of $D_n$ on $\lambda$ is unknown, we assume
that the gradient of the mean AoI with respect to
$\lambda$ admits the representation
\begin{equation}
\frac{\partial}{\partial \lambda}\, \E_{\lambda}[A]
= g(\lambda;\bm{\theta}_{\lambda}),
\label{eq:g(lambda,theta)-def}
\end{equation}
where $g: \mathbb{R}\times \mathbb{R}^d \to \mathbb{R}$ (for some $d\in{\mathbb N}$) denotes a {\it known} function and 
$\bm{\theta}_{\lambda} \in \mathbb{R}^d$ denotes a parameter dependent on $\lambda$.
Thus, the unknown system characteristics enter the optimization problem only through the finite-dimensional vector $\bm{\theta}_\lambda$.
Moreover, we assume that $\bm{\theta}_{\lambda}$ can be consistently estimated
from the observed sequences of informative packets, namely from the
samples of $(G_{\ell}^{\dag}, D_{\ell}^{\dag})_{\ell = 1,2,\ldots}$. This enables
data-driven adjustment of $\lambda$ based solely on monitor-side
observations.
\end{enumerate}
The assumptions introduced above are satisfied by a broad class of models; several representative examples are discussed in Section \ref{ssec:examples}.

In our setup, the receiver has access to the exact timestamps $\alpha_{\ell}^{\dag}$ and $\beta_{\ell}^{\dag}$ of the $\ell$-th informative packet upon its reception. Equivalently, after the $L$-th informative packet has been received, the receiver can reconstruct the subsequences $(G_{\ell}^{\dag})_{\ell=1,2,\ldots,L}$ and $(D_{\ell}^{\dag})_{\ell=1,2,\ldots,L}$ of inter-generation times and system delays.
Let $\ell^{\step{k}}$ {(for $k=1,2,\ldots$)} denote the $k$-th decision epoch. Upon reception of the $\ell^{\step{k}}$-th packet, the receiver selects the value $\lambda^{\step{k}}$ of the generation parameter, which is then used by the sender until the next decision epoch $\ell^{\step{k+1}}$. The sequence $(\ell^{\step{k}})_{k=0,1,\ldots}$ is strictly increasing, i.e.,
\[
0 \le \ell^{\step{k}} < \ell^{\step{k+1}}, \qquad k=0,1,\ldots,
\]
with $\ell^{\step{0}} = 0$ for convenience. The initial value $\lambda^{\step{0}}$ is chosen at time $t=0$.

It is then verified that the $k$-th decision epoch $\ell^{\step{k}}$
corresponds to time $t=t^{\step{k}}$ defined as
\[
t^{\step{k}} := \beta_{\ell^{\step{k}}}^{\dag},
\;\;
k = 0,1,\ldots.
\]
Let $\overline{A}_{\tau}$ ($\tau > 0$) denote the cumulative AoI over a time interval $[0,\tau)$:
\[
\overline{A}_{\tau} = \int_0^{\tau} A_t\, \dd t.
\]
By definition, we have
\begin{align*}
\overline{A}_{t^{\step{m}}}
&=
\int_0^{t^{\step{1}}} A_t\, \dd t
+
\sum_{k=1}^{m-1}
\int_{t^{\step{k}}}^{t^{\step{k+1}}} A_t\, \dd t
=
\sum_{k=0}^{m-1} \overline{A}^{\step{k}},
\end{align*}
where
\[
\overline{A}^{\step{0}} := \int_0^{t^{\step{1}}} A_t\, \dd t,
\quad
\overline{A}^{\step{k}}
:=
\int_{t^{\step{k}}}^{t^{\step{k+1}}} A_t\, \dd t,
\;\;
k = 1,2,\ldots.
\]
Therefore, assuming $t^{\step{k}} \to \infty$ as $k \to \infty$, the time-averaged AoI in this setting is given by
\[
\limsup_{\tau \to \infty}\frac{1}{\tau}\int_0^{\tau} A_t\, \dd t
=
\limsup_{\tau \to \infty}
\frac{1}{\tau}\sum_{k=0}^{M_{\tau}-1} 
\overline{A}^{\step{k}},
\]
where $M_t$ is as defined in \eqref{eq:M_t-def}.

From this viewpoint, we consider a dynamic optimization problem in which the
objective combines the long-run average AoI with a cost associated with the
generation rate. Specifically, we define the long-run average cost as
\begin{align}
Q:=
\limsup_{\tau \to \infty}
\frac{1}{\tau}\sum_{k=0}^{M_{\tau}-1}
\left\{
\overline{A}^{\step{k}}
+
(t^{\step{k+1}} - t^{\step{k}})\,\kappa(\lambda^{\step{k}})
\right\},
\label{eq:objective}
\end{align}
where $\kappa(\lambda)$ represents the cost per unit time associated with
using the generation parameter $\lambda$.
As mentioned above, the subsequences
$(G_{\ell}^{\dag})_{\ell=1,2,\ldots,\ell^{\step{k}}}$
and $(D_{\ell}^{\dag})_{\ell=1,2,\ldots,\ell^{\step{k}}}$ of inter-generation
times and system delays are the only information available when determining
$\lambda^{\step{k}}$.
Throughout this paper, we assume that the generation cost $\kappa(\lambda)$ is
a non-decreasing convex function of $\lambda$. A particular example is a
linear function $\kappa(\lambda) = C\lambda$, where $C>0$ controls the
tradeoff between the generation rate and the AoI.

When $\kappa$ is differentiable, the structural representation
\eqref{eq:g(lambda,theta)-def} also yields the gradient of the corresponding
stationary long-run objective. In particular,
\begin{equation}
\frac{\partial}{\partial\lambda}
\left\{
\E_{\lambda}[A] + \kappa(\lambda)
\right\}
=
g(\lambda;\bm{\theta}_{\lambda}) + \kappa'(\lambda).
\notag
\end{equation}
Thus, the unknown component of the gradient remains confined to the
finite-dimensional vector $\bm{\theta}_{\lambda}$, while the contribution
from the generation cost is known.

\subsection{Examples}
\label{ssec:examples}

In this subsection, we present several examples of the general problem setting introduced in Section~\ref{ssec:problem}. These examples illustrate that the chosen framework is sufficiently general to encompass a range of relevant scenarios.

\begin{itemize}
\item \textit{Example 1: FCFS dedicated channel with unknown service-time distribution}\\
Consider a first-come first-served (FCFS) single-server queue where the sensor generates packets according to a Poisson process with rate $\lambda$ and service times are independent and identically distributed (i.i.d.), i.e., the FCFS M/GI/1 queue.
Suppose that the monitor is to optimize the packet generation rate $\lambda$ of the sensor, without knowledge on
the service time distribution.

Although the exact mean AoI $\E_{\lambda}[A]$ as a function of $\lambda$ is not available due to the absence of the information about the service-time distribution, we can still utilize structural knowledge from analytical results.
More specifically, let $H$ denote a generic random variable for service times and let $f_H^*(s) := \E[e^{-sH}]$ ($s \geq 0$) denote the corresponding Laplace-Stieltjes transform (LST).
Under the stability condition $\rho := \lambda \E[H] < 1$, the mean AoI in this model is given by \cite{Inoue2019}
\begin{equation}
\E_{\lambda}[A] 
= 
\frac{\lambda\, \E[H^2]}{2(1-\lambda\,\E[H])}
+
\E[H]+ \frac{1-\lambda\,\E[H]}{\lambda\,f_H^{\star}(\lambda)}.
\label{eq:EA}
\end{equation}
We can then verify that the gradient of $\E_{\lambda}[A]$
is expressed in the form \eqref{eq:g(lambda,theta)-def} with
\begin{equation}
g(\lambda;{\boldsymbol \theta}_\lambda)=
\frac{\E[H^2]}{2(1-\lambda\,\E[H])^2}
-\frac{f_H'(\lambda)}{\lambda\{f_H^*(\lambda)\}^2}
-\frac{1}{\lambda^2\,f_H^*(\lambda)}
+\frac{\E[H]f_H'(\lambda)}{\{f_H^*(\lambda)\}^2},
\label{eq:psi_grad-0}
\end{equation}
where
\begin{align}
f_H'(s) &:= \frac{\dd}{\dd s}f_H^*(s),
\label{eq:f_H'-def}
\\
{\boldsymbol \theta}_\lambda&\equiv\left(\E[H],\E[H^2],f_H^*(\lambda),  f_H'(\lambda)\right).
\nonumber
\end{align}
Observe that the first two parameters in ${\boldsymbol\theta}_{\lambda}$ do not
depend on the generation rate $\lambda$, while the latter two vary if the rate
changes. Thus, although the service-time distribution is unknown, the AoI
gradient depends only on the finite-dimensional vector
${\boldsymbol\theta}_{\lambda}$, whose components can be estimated from
informative-packet observations.
\hfill $\diamond$

\item \textit{Example 2: FCFS with background sources of unknown generation rate}\\
Consider an FCFS M+M/GI/1 queue, i.e., an M/GI/1 queue with a background M/GI stream. The tagged stream (resp.\ background stream) generates packets according to a Poisson process with rate $\lambda$ (resp.\ $\lambda_{\bg}>0$). Service times of packets are i.i.d.\ regardless of the source of the packets.
Suppose that the monitor knows the service time distribution but does not have access to information about arrival or departure events of packets from the background stream.
As a result, the generation rate $\lambda$ of the tagged stream should be optimized without any knowledge of $\lambda_{\bg}$.

Let $H$ denote a generic random variable for service times and $f_H^*(s)$ denote its LST as before. 
Under the stability condition
\begin{equation}\label{eq:stability}
\rho:=(\lambda+\lambda_{\bg})\E[H] < 1,
\end{equation}
the mean AoI $\E_{\lambda}[A]$ of the tagged stream is given by \cite{Inoue2024,Inoue2025}:
\begin{equation}\label{eq:EA_bg}
\E_{\lambda}[A] 
= 
\frac{(\lambda+\lambda_{\bg}) \E[H^2]}
{2\{1-(\lambda + \lambda_{\bg}) \E[H]\}}
+
\E[H]
+
\frac{\lambda_{\bg}\E[H]}{\lambda}
+
\frac{1-(\lambda+\lambda_{\bg})\E[H]}
{\lambda \zeta(\lambda,\lambda_\bg)},
\end{equation}
where $\zeta(\lambda,\lambda_\bg)$ denotes the unique solution of
\begin{equation}
\label{eq:Ex2FP}
x  = f_H^*(\lambda + \lambda_{\bg} - \lambda_{\bg} x),
\quad
0 \leq x \leq 1.
\end{equation}
The gradient is then given by
\begin{align*}
g(\lambda;{\boldsymbol \theta}_\lambda)
&=
\frac{\E[H^{2}]}{2\left(1-(\lambda+\lambda_{\bg})\E[H]\right)^{2}}
-\frac{\lambda_{\bg}\E[H]}{\lambda^{2}}
-\frac{\E[H]}{\lambda\,\zeta(\lambda,\lambda_\bg)}
\\
&\qquad\qquad{}-\frac{1-(\lambda+\lambda_{\bg})\E[H]}{\lambda^2 \zeta(\lambda,\lambda_\bg)^{2}}\,
\left(\zeta(\lambda,\lambda_\bg)+\lambda \zeta'(\lambda,\lambda_\bg)\right),
\\
\zeta'(\lambda,\lambda_\bg) 
&:= 
\frac{\partial}{\partial \lambda}
\zeta(\lambda, \lambda_\bg),
\end{align*}
where the parameter vector is  
\begin{align*}
{\boldsymbol \theta}_\lambda\equiv \left(\E[H],\E[H^2],\lambda_\bg,\zeta(\lambda,\lambda_\bg),  \zeta'(\lambda,\lambda_\bg)\right).
\end{align*}

Note that $\zeta(\lambda,\lambda_{\bg})$ can also be characterized in terms of the length $B_{\bg}$ of a busy period in an M/GI/1 queue with only background traffic load. Recall that, for each $s \geq 0$, the LST $f_{B_{\bg}}^*(s)$ of $B_{\bg}$ is given by the unique solution $x \in [0,1]$ of the fixed-point equation
\[
x = f_H^*(s + \lambda_{\bg} - \lambda_{\bg} x), \qquad 0 \leq x \leq 1.
\]
Therefore, $\zeta(\lambda,\lambda_\bg)$ and $\zeta'(\lambda,\lambda_\bg)$ are expressed as
\begin{align}
\zeta(\lambda,\lambda_\bg) &= f_{B_{\bg}}^*(\lambda),\qquad \zeta'(\lambda,\lambda_\bg)
= \frac{\dd}{\dd \lambda}f_{B_{\bg}}^*(\lambda)
\label{eq:zeta-by-f_B_bg}
\end{align}
under the condition \eqref{eq:stability}. Hence, even though the background
arrival rate $\lambda_{\bg}$ is unobserved, the AoI gradient admits the
finite-dimensional representation \eqref{eq:g(lambda,theta)-def}, with the
unknown quantities contained in ${\boldsymbol\theta}_{\lambda}$ and
estimable from observations of the tagged stream.
\hfill $\diamond$ 

\item \textit{Example 3: Preemptive LCFS System with unknown service-time distribution}\\
Consider a preemptive last-come first-served (LCFS) {M/GI/1} queue, where packets are
generated according to a Poisson process with rate $\lambda$ and service times
are i.i.d. When a new packet is generated while another packet is in service,
the older one is discarded and the newer one immediately starts service. In this
case, the monitor can observe only the service times of successfully received
packets. The distribution of these observed service times is generally biased.
Therefore, if the full service-time distribution is not known a priori, the
generation rate $\lambda$ should be optimized based on such partial
observations. 

In this model, the mean AoI is given by \cite{Inoue2019} 
\[
\E_{\lambda}[A] = \frac{1}{\lambda f_H^*(\lambda)},
\]
where $f_H^*(s)$ denotes the LST of the full service-time distribution.
The gradient is given by
\begin{align*}
g(\lambda; {\boldsymbol \theta}_\lambda)
&=
-\frac{f_H^*(\lambda) + \lambda f_H'(\lambda)}{\lambda^2 \{f_H^*(\lambda)\}^2},
\\
{\boldsymbol \theta}_\lambda 
&\equiv 
\left(f_H^*(\lambda), f_H'(\lambda)\right),
\end{align*}
with $f_H'(s)$ defined as \eqref{eq:f_H'-def}. Thus, despite having access only
to the service times of successfully received packets, the relevant
quantities in the AoI gradient can be estimated from these informative
observations, providing an instance of the proposed structure-aware framework
with partial observations.
\hfill $\diamond$

\item \textit{Example 4: Generate-at-will system with threshold waiting \cite{Sun2017}}\\
Consider a generate-at-will system in which 
the source has access to the idle/busy state of the server through ACK signals and it generates a new packet
only after the previous one has been delivered. Let $H_n$ denote the
service time of the $n$-th packet. After observing $H_n$, 
the source waits for an additional duration $Z_n := \xi_{\lambda}(H_n)
$, where 
\begin{equation}\label{eq:xi}
\xi_{\lambda}(x)
:=
\left\{
\begin{array}{l@{\quad}l}
\displaystyle\frac{1}{\lambda} - x, & x < \displaystyle \frac{1}{\lambda},
\\[2ex]
0, & x \geq \displaystyle \frac{1}{\lambda}.
\\
\end{array}
\right.
\end{equation}
In this case, the server is always idle at packet generation epochs, so every delivered packet is informative. We therefore have
\[
G_n^{\dag} = H_{n-1} + \xi_{\lambda}(H_{n-1}),
\qquad
D_n^{\dag} = H_n.
\]
Suppose that $(H_n)_{n=0,1,\ldots}$ are i.i.d.\ with mean $0 < \E[H] < \infty$, and let $H$ denote a generic copy of $H_n$.
For this threshold policy, the mean AoI is given by \cite{Sun2017}
\begin{align*}
\E_{\lambda}[A]
=
\frac{\E_{\lambda}[(G^{\dag})^2]}{2\E_{\lambda}[G^{\dag}]}
+
\E[H],
\end{align*}
where $G^{\dag}$ denotes a generic random variable corresponding to $G_n^{\dag}$. Consequently,
\[E_{\lambda}[G^{\dag}]
=
\E\!\left[H+\xi_{\lambda}(H)\right], \qquad
\E_{\lambda}[(G^{\dag})^2]
=
\E\!\left[(H+\xi_{\lambda}(H))^2\right].
\]
If $H$ has an absolutely continuous distribution, the gradient of the mean AoI $\E_{\lambda}[A]$ takes the form \eqref{eq:g(lambda,theta)-def}, where
\begin{align}
g(\lambda;{\boldsymbol \theta}_\lambda)
&=
\frac{
\Pr(H<1/\lambda)
}{
2\lambda^2 \E_{\lambda}[G^{\dag}]^2
}
\left(
\E_{\lambda}[(G^{\dag})^2]
-\frac{2}{\lambda}\E_{\lambda}[G^{\dag}]
\right)\notag
\\ \label{eq:gEx4}
&=
\frac{
\Pr(D^{\dag}<1/\lambda)
}{
2\lambda^2 \E_{\lambda}[G^{\dag}]^2
}
\left(
\E_{\lambda}[(G^{\dag})^2]
-\frac{2}{\lambda}\E_{\lambda}[G^{\dag}]
\right);
\end{align}
here 
\begin{equation}
{\boldsymbol \theta}_\lambda
\equiv 
\Bigl(
\E_{\lambda}[G^{\dag}],\,
\E_{\lambda}[(G^{\dag})^2],\,
\Pr(D^{\dag}<1/\lambda)
\Bigr).
\label{eq:thetaEx4}
\end{equation}
The vector ${\boldsymbol\theta}_{\lambda}$ is directly estimable from the
observed inter-generation times and system delays of informative packets.
Thus, this generate-at-will system also fits the general framework: the gradient can be represented explicitly in terms of a finite-dimensional parameter, while its
required components can be learned online from monitor-side observations.
\hfill $\diamond$

\end{itemize}
%

\section{Proposed Method and its Convergence}
\label{sec:alg}

\begin{figure} 
    \centering 
    \begin{tikzpicture}[>=Stealth,node distance=1.0cm]
        \node[draw,rounded corners,fill=lightgray,minimum width=2.4cm,minimum height=0.9cm] (sys) {\small Generate packets with $\lambda^{\step{k}}$};
        \node[draw,rounded corners,fill=blue!10,below=0.9cm of sys,minimum width=2.6cm,minimum height=0.9cm] (data) {\small Data $(G_\ell^{\dag},D_\ell^{\dag})_{\ell=\ell^{\step{k}}+1}^{\ell^{\step{k+1}}}$};
        \node[draw,rounded corners,fill=blue!10,right=1.4cm of data,minimum width=2.4cm,minimum height=0.9cm] (est) {\small Estimate $\hat{\theta}^{\step{k}}$};
        \node[draw,rounded corners,fill=red!10,right=1.4cm of est,minimum width=2.4cm,minimum height=0.9cm] (grad) {\small $g(\lambda^{\step{k}};\hat\theta^{\step{k}})$};
        \node[draw,rounded corners,fill=red!10,above=0.9cm of grad,minimum width=2.4cm,minimum height=0.9cm] (upd) {\small SGD update $\lambda^{\step{k+1}}$};
        
        \draw[->,thick] (sys) -- (data);
        \draw[->,thick] (data) -- (est);
        \draw[->,thick] (est) -- (grad);
        \draw[->,thick] (grad) -- (upd);
        \draw[->,thick] (upd) -- (sys);
    \end{tikzpicture}
    
    \caption{Schematic of the proposed algorithm. Iterations continue until a predefined stopping condition is satisfied. }
    \label{fig:SGD_overview}
\end{figure}
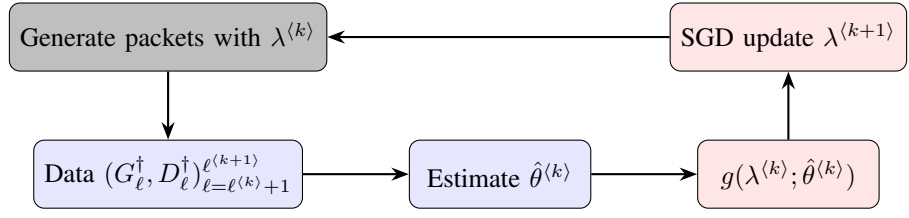

{In this section, we develop a stochastic gradient descent (SGD) algorithm to minimize the objective function $Q$ given by \eqref{eq:objective}; see Fig.\ \ref{fig:SGD_overview} for a pictorial illustration of the proposed procedure. 
The key idea is to exploit the analytical gradient representation in
\eqref{eq:g(lambda,theta)-def}: rather than estimating the AoI objective
itself, the algorithm estimates the finite-dimensional parameter vector
$\bm{\theta}_{\lambda}$ that determines the gradient and uses the resulting
estimate to update the generation rate.
Recall that $\lambda^{\step{k}}$ ($k=0,1,\ldots$) denotes the generation rate
selected at the $k$-th decision epoch. Over the interval between the $k$-th and $(k+1)$-st decision epochs, the generation rate is held fixed at $\lambda^{\step{k}}$. Based on the informative packets received during this interval, we estimate
the finite-dimensional parameter vector $\bm{\theta}_{\lambda^{\step{k}}}$
that enters the analytical gradient representation \eqref{eq:g(lambda,theta)-def}:
\begin{equation}
\hat{\bm{\theta}}^{\step{k}} 
:= 
\bm{\vartheta}_{\lambda^{\step{k}}}\big(G_{\ell^{\step{k}}+1}^{\dag},
G_{\ell^{\step{k}}+2}^{\dag},\ldots,
G_{\ell^{\step{k+1}}}^{\dag},
D_{\ell^{\step{k}}+1}^{\dag},D_{\ell^{\step{k}}+2}^{\dag},\ldots,D_{\ell^{\step{k+1}}}^{\dag}\big),
\label{eq:hat-theta-def}
\end{equation}
where the function $\bm{\vartheta}_{\lambda^{\step{k}}}$ on the right-hand
side denotes an estimator of $\bm{\theta}_{\lambda^{\step{k}}}$ based on the
informative-packet observations collected during the $k$-th decision
interval.

We assume that the parameter $\lambda^{\step{k}}$ is selected from a closed interval $\mathcal{L} := [\lambda_{\min}, \lambda_{\max}]$ for some $0<\lambda_{\min}<\lambda_{\max}<\infty$.
The proposed algorithm is then defined through
\begin{equation}
\lambda^{\step{k}}:=\Pi_{\mathcal{L}}\left[\lambda^{\step{k-1}}-\gamma^{\step{k}}\{g(\lambda^{\step{k-1}};\hat{\bm{\theta}}^{\step{k-1}})+\kappa'(\lambda^{\step{k-1}})\}\right], 
\quad
k = 1,2,\ldots,
\label{eq:update-lambda}
\end{equation}
where $\gamma^{\step{k}} > 0$ denotes a step size and
$\Pi_{\mathcal{L}}[\cdot]$ denotes the projection of $\lambda^{\step{k}}$ onto the allowable interval $\mathcal{L}$:
\[
\Pi_{\mathcal{L}}[x] := \max\{\min\{x,\lambda_{\max}\},\lambda_{\min}\}.
\]
In our SGD algorithm, we set the step size as, for constants $c_1>0$ and $\alpha\in(0.5,1]$,
\begin{equation}
\gamma^{\step{k}} := c_1 k^{-\alpha},
\quad
k = 1,2,\ldots.
\label{eq:step-size}
\end{equation}
We further set the number of samples $n^{\step{k}} := \ell^{\step{k+1}}-\ell^{\step{k}}$ used in the $k$-th interval as, for constants $c_2 > 0$ and $\beta > 2(1-\alpha)$,
\begin{equation}
n^{\step{k}} := c_2 k^{\beta}, \quad
k = 1,2,\ldots.
\label{eq:sample-size}
\end{equation}

Intuitively, the proposed algorithm holds the generation rate fixed at
$\lambda^{\step{k}}$ over the $k$-th decision interval and uses the informative
packets collected during this interval to estimate the structural quantities
$\bm{\theta}_{\lambda^{\step{k}}}$ entering the analytical gradient. Although
the observations are generated by an evolving queueing system and may exhibit
intricate temporal dependence, the algorithm does not attempt to learn the
AoI objective as a black box. Instead, it estimates only the finite-dimensional
quantities required by the analytical gradient and uses
$g(\lambda^{\step{k}};\hat{\bm{\theta}}^{\step{k}})$ as a data-driven gradient
estimate. The update rule \eqref{eq:update-lambda} then adjusts the generation
rate in the direction opposite to the estimated gradient of the total
objective. The projection operator enforces feasibility of the iterates,
while the diminishing step size and increasing sample size control the
estimation error over successive iterations.}

{The main technical challenge is that the gradient estimate is constructed
from observations generated by the queueing system itself. Thus, unlike
standard stochastic approximation settings with exogenous noise, the
estimation error is coupled with the current generation rate and inherits
the intricate temporal dependence of the underlying queueing process.
We capture this dependence through a uniform mean squared error condition
on the estimator $\hat{\bm{\theta}}^{\step{k}}$.

In what follows, we show that, under suitable regularity assumptions, the
proposed algorithm generates a sequence of update rates
$(\lambda^{\step{k}})_{k=0,1,\ldots}$ that converges to the optimal update
rate $\lambda_*$.

\begin{assumption}\label{ass:regularity}
\begin{enumerate}
\item[(i)] The stationary mean AoI $\E_\lambda[A]$ is well defined, finite, and continuously differentiable on $\lambda \in \mathcal{L} = [\lambda_{\min}, \lambda_{\max}]$.

\item[(ii)] The gradient of $\E_\lambda[A]$ admits the representation \eqref{eq:g(lambda,theta)-def} for a known function $g$ and a parameter vector $\bm{\theta}_{\lambda} \in \Theta$, where $\Theta \subseteq \mathbb{R}^d$ denotes the parameter space.

\item[(iii)] There exists a unique $\lambda_* \in \mathcal{L}$ such that
\[
g(\lambda_*;\bm{\theta}_{\lambda_*}) + \kappa'(\lambda_*) = 0.
\]

\item[(iv)] There exists a constant $C_1 > 0$ such that, for all $\bm{\theta}, \bm{\theta}' \in \Theta$ and all $\lambda\in{\mathcal L}$,
\begin{equation} \label{eq:Lips}
|g(\lambda;\bm{\theta}) - g(\lambda;\bm{\theta}')| \le C_1 \|\bm{\theta} - \bm{\theta}'\|.
\end{equation}
\item[(v)] There exists a constant $C_2 > 0$ such that,  for all $\lambda\in{\mathcal L}$,
\begin{equation}\label{eq:theta_mixing}
\mathbb{E}_{\lambda^{\step{k}}}\!\left[ \|\bm{\theta}_{\lambda^{\step{k}}} - \hat{\bm{\theta}}^{\step{k}}\|^2\, \Big|\, \,
\mathcal{F}^{\step{k-1}}, \lambda^{\step{k}} \right]
\le \frac{C_2}{n^{\step{k}}}, \qquad k=1,2,\ldots,
\end{equation}
where $\mathcal{F}^{\step{k-1}}$ denotes the sigma algebra generated by the observations in all iterations before the $k$-th decision epoch.
\end{enumerate}
\end{assumption}

Parts (i) and (ii) of Assumption~\ref{ass:regularity} are mild and can be verified directly whenever a closed-form expression for the mean AoI is available. Since $\kappa(\lambda)$ is assumed to be convex, part (iii) is satisfied,
for example, if $\E_\lambda[A]$ is strictly convex on $\mathcal{L}$.  We show in Section~\ref{sec:application} that Examples 1 and 2 in Section~\ref{sec:model} exhibit a strongly convex structure in $g$, so that this condition holds in those cases.

Under Assumption~\ref{ass:regularity} (i)--(iii), define the optimal
update rate as
\begin{equation}
\lambda_* = \arg\min_{\lambda \in \mathcal{L}}
\bigl(\E_{\lambda}[A] + \kappa(\lambda)\bigr).
\label{eq:lambda-opt-def}
\end{equation}
This is the objective obtained when the generation rate is held fixed at
$\lambda$. Moreover, $\lambda_*$ is equivalently characterized as the
unique solution in $\mathcal{L}$ of
\begin{equation}
g(\lambda_*;\bm{\theta}_{\lambda_*}) + \kappa'(\lambda_*) = 0.
\label{eq:zero_gradient}
\end{equation}

Assumption~\ref{ass:regularity} (iv) requires the gradient representation to
be Lipschitz continuous in the estimated parameter vector. Assumption~\ref{ass:regularity}
(v) imposes a uniform $O(1/n^{\step{k}})$ mean squared error bound on the
estimator of the structural parameter vector. Establishing this bound is
where the temporal dependence induced by the underlying queueing process
enters the analysis. In Section~\ref{sec:application}, we verify this
condition for several important special cases under appropriate moment
conditions on the service-time distribution.

We are now in a position to state and prove the main result of this paper.

\begin{theorem}\label{thm:SGD_convergence_2}
Suppose that Assumption~\ref{ass:regularity} holds. Then the proposed SGD algorithm, with step size $\gamma^{\step{k}}$ and sample size $n^{\step{k}}$ given in \eqref{eq:step-size} and \eqref{eq:sample-size}, with $\beta > 2(1-\alpha)$,
generates a sequence of update rates $(\lambda^{\step{k}})_{k=0,1,\ldots}$ that converges almost surely to the minimizer $\lambda_*$ of \eqref{eq:objective}, i.e.,
\begin{align}
\lambda^{\step{k}} \xrightarrow{\text{\rm a.s.}} \lambda_*.
\label{eq:lambda_convergence}
\end{align}
\end{theorem}}

The proof combines a bias--variance analysis of a projected stochastic
approximation scheme with the mean squared estimation bound in
Assumption~\ref{ass:regularity}(v). The gradient estimate is not, in general,
unbiased because it is obtained by applying the nonlinear function $g$ to an
estimated parameter vector. Moreover, the estimation error is generated by
the queueing process and therefore inherits its temporal dependence. The
increasing sample size $n^{\step{k}}$ ensures that this bias becomes
summable when weighted by the diminishing step size $\gamma^{\step{k}}$,
while the same estimation bound yields a uniform variance bound. These
properties allow us to invoke standard stochastic approximation arguments
to establish almost sure convergence.

\begin{proof}
{Let $B^{\step{k}}$ denote the estimation bias of the gradient at iteration $k$, defined by
\begin{equation}\label{eq:bias_g}
B^{\step{k}}
:=
g(\lambda^{\step{k}}; \bm{\theta}_{\lambda^{\step{k}}})
-
g(\lambda^{\step{k}};\hat{\bm{\theta}}^{\step{k}}),
\qquad k = 1,2,\ldots.
\end{equation}
Since both $\lambda^{\step{k}}$ and $\hat{\bm{\theta}}^{\step{k}}$ are random, $B^{\step{k}}$ is a random variable.
The iteration \eqref{eq:update-lambda} can then be rewritten as
\begin{align*}
\lambda^{\step{k+1}}
&=
\lambda^{\step{k}}
-\gamma^{\step{k}}
\left\{
g\!\left(\lambda^{\step{k}};\bm{\theta}_{\lambda^{\step{k}}}\right)
+\kappa'(\lambda^{\step{k}})
\right\}
+\gamma^{\step{k}} B^{\step{k}}
+ Y^{\step{k}},
\end{align*}
where the `projection residual' $Y^{\step{k}}$ is defined via
\[
Y^{\step{k}}
:=
\Pi_{\mathcal{L}}\!\left[\lambda^{\step{k}} - \gamma^{\step{k}} \hat{G}^{\step{k}}\right]
-
\left(\lambda^{\step{k}} - \gamma^{\step{k}} \hat{G}^{\step{k}}\right),
\quad
\hat{G}^{\step{k}}
:=
g\!\left(\lambda^{\step{k}};\hat{\bm{\theta}}^{\step{k}}\right)
+\kappa'(\lambda^{\step{k}}).
\]
The term $Y^{\step{k}}$ is nonzero only when the unconstrained update leaves the feasible set $\mathcal{L}$, and in that case it represents the correction enforcing projection onto $\mathcal{L}$.

Let $\mathcal{F}^{\step{k-1}}$ denote the sigma algebra generated by the observations in all iterations before the $k$-th decision epoch, i.e., 
\[
\mathcal{F}^{\step{k-1}}:=\sigma\left((D_1^{\dag},G_1^{\dag}),(D_2^{\dag},G_2^{\dag})\ldots,(D^\dag_{\ell^{\step{k}}},G^\dag_{\ell^{{\step{k}}}}), \lambda^{\step{k-1}}\right).
\]
Recall that $\lambda^{\step{k}}$ is determined by $\lambda^{\step{k-1}}$ and the data 
$((G_{\ell}^{\dag}, D_{\ell}^{\dag}))_{\ell \in \{\ell^{\step{k-1}}+1,\ell^{\step{k-1}}+2,\ldots,\ell^{\step{k}}\}}$ observed in iteration $k-1$ (cf.\ \eqref{eq:hat-theta-def} and \eqref{eq:update-lambda}). 
Furthermore, we define $b^{\step{k}}$ as the absolute mean bias:
\[
b^{\step{k}} := \left| \mathbb{E}\!\left[ B^{\step{k}} \, \Big|\, \lambda^{\step{k}}, \mathcal{F}^{\step{k-1}} \right] \right|.
\] 
Assumption~\ref{ass:regularity} (iii) ensures that 
the optimal $\lambda_*$ is the unique solution of \eqref{eq:zero_gradient}.
Therefore, owing to \cite[Page 127]{KY2003},
the almost sure convergence \eqref{eq:lambda_convergence} holds if the following conditions are satisfied:
\begin{enumerate}
\item[(i)] The step size sequence satisfies
\begin{align*}
\sum_{k=1}^\infty \gamma^{\step{k}}=\infty, 
\quad 
\sum_{k=1}^\infty (\gamma^{\step{k}})^2 < \infty.
\end{align*} 
\item[(ii)] There exists a constant $M<\infty$ such that for any $\lambda^{\step{k}}\in[\lambda_{\min},\lambda_{\max}]$ and $n^{\step{k}}$,
\begin{equation}\label{eq:variance_bound}
\mathbb{E}\!\left[\left(B^{\step{k}}\right)^2 \, \Big|\,  \lambda^{\step{k}}, \mathcal{F}^{\step{k-1}} \right] \le M.
\end{equation}
\item[(iii)] The sample size sequence $(n^{\step{k}})_{k=1,2,\ldots}$ satisfies, almost surely,
\begin{equation}\label{eq:bias}
\sum_{k=1}^\infty \gamma^{\step{k}}|b^{\step{k}}|<\infty .
\end{equation}
\end{enumerate} 
It is readily verified that $\gamma^{\step{k}}$ given by \eqref{eq:step-size} satisfies (i).
Therefore, the proof requires verifying (ii) and (iii).

First of all, from Assumption~\ref{ass:regularity} (iv) and the definition
\eqref{eq:bias_g} of $B^{\step{k}}$, we have
\begin{align*}
\mathbb{E}\!\left[ \left(B^{\step{k}}\right)^2 \, \Big|\,
\mathcal{F}^{\step{k-1}}, \lambda^{\step{k}} \right]
&\le C_1^2 \,
\mathbb{E}\!\left[
\left\|
\bm{\theta}_{\lambda^{\step{k}}}
-
\hat{\bm{\theta}}^{\step{k}}
\right\|^2
\, \Big|\,
\mathcal{F}^{\step{k-1}}, \lambda^{\step{k}}
\right].
\end{align*}
Assumption~\ref{ass:regularity} (v) thus implies
\begin{align*}
\mathbb{E}\!\left[
\left(B^{\step{k}}\right)^2 \, \Big|\,
\mathcal{F}^{\step{k-1}},\lambda^{\step{k}}
\right]
\leq
\frac{C_1^2 C_2}{n^{\step{k}}}.
\end{align*}
Since $n^{\step{k}}$ is an increasing sequence of positive integers, we
conclude that \eqref{eq:variance_bound} holds with
$M=C_1^2C_2$.
Furthermore, by the Cauchy--Schwarz inequality,
\begin{align*}
|b^{\step{k}}|
&=
\left|
\mathbb{E}\!\left[
B^{\step{k}} \, \Big|\,
\mathcal{F}^{\step{k-1}}, \lambda^{\step{k}}
\right]
\right|\\
&\le
\sqrt{
\mathbb{E}\!\left[
\left(B^{\step{k}}\right)^2 \, \Big|\,
\mathcal{F}^{\step{k-1}}, \lambda^{\step{k}}
\right]
}
\leq
\frac{C_1\sqrt{C_2}}{\sqrt{n^{\step{k}}}}.
\end{align*}
Therefore, we obtain \eqref{eq:bias} from \eqref{eq:step-size} and
\eqref{eq:sample-size}, as follows:
\begin{align*}
\sum_{k=1}^\infty \gamma^{\step{k}}|b^{\step{k}}|
&\leq
C_1\sqrt{C_2}
\sum_{k=1}^\infty
\frac{c_1 k^{-\alpha}}{\sqrt{c_2 k^{\beta}}}
=
\frac{C_1\sqrt{C_2}\,c_1}{\sqrt{c_2}}
\sum_{k=1}^\infty
\frac{1}{k^{\alpha+\beta/2}}
<\infty,
\end{align*}
recalling the condition $\beta > 2(1-\alpha)$.
}
\end{proof}

\begin{remark}
For ease of presentation, we adopt the specific choices \eqref{eq:step-size} and \eqref{eq:sample-size} for the step size $\gamma^{\step{k}}$ and sample size $n^{\step{k}}$. These particular forms are not essential for the convergence result of Theorem~\ref{thm:SGD_convergence_2}. The proof only requires that the sequences satisfy
\[
\sum_{k=1}^{\infty} \gamma^{\step{k}} = \infty,
\qquad
\sum_{k=1}^{\infty} (\gamma^{\step{k}})^2 < \infty,
\qquad
\sum_{k=1}^{\infty} \frac{\gamma^{\step{k}}}{\sqrt{n^{\step{k}}}} < \infty.
\]
Under these conditions, Theorem~\ref{thm:SGD_convergence_2} remains valid for more general choices of $\gamma^{\step{k}}$ and $n^{\step{k}}$.
\end{remark}

\section{Application of the Framework to Examples 1--4}
\label{sec:application} 

In this section, we demonstrate the flexibility of our framework by detailing how the proposed procedure can be implemented for Examples~1--4 in Section~\ref{sec:model}. We discuss both practical and theoretical aspects of these implementations, with particular emphasis on verifying the assumptions required for convergence of the algorithm, as established in the preceding section, across the different settings.


\subsection{Example 1: FCFS dedicated channel with unknown service-time distribution}\label{sec:application_A}

Consider the FCFS M/GI/1 queue in Example 1 of 
Section~\ref{ssec:examples}.
Recall that the mean AoI $\E_{\lambda}[A]$ in this case is given by \eqref{eq:EA}.
We assume $\E[H^2] < \infty$ so that $\E[A]$ is finite for $\lambda \in (0,1/\E[H])$.
We first show that $\E_{\lambda}[A]$ is a convex function of the generation rate $\lambda$; we provide its proof in Appendix \ref{appendix:FD_convex}.
\begin{theorem}\label{thm:FD_convex}
The mean AoI $\E_{\lambda}[A]$ given by \eqref{eq:EA} is a strictly convex function of $\lambda \in \left(0,1/\E[H]\right)$.
\end{theorem}

We next construct the parametric gradient estimator. Recall that the gradient in \eqref{eq:psi_grad-0} can be written as
\begin{align*}
g(\lambda;{\boldsymbol \theta}_\lambda)=
\frac{\E[H^2]}{2(1-\lambda\,\E[H])^2}
-\frac{f_H'(\lambda)}{\lambda\{f_H^*(\lambda)\}^2}
-\frac{1}{\lambda^2\,f_H^*(\lambda)}
+\frac{\E[H]f_H'(\lambda)}{\{f_H^*(\lambda)\}^2}.
\end{align*}

In this setting, all packets are informative since there is no overtaking, so that $G_i^\dag=G_i$ and $D_i^\dag=D_i$ for all $i = 1,2,\ldots$. Moreover, observing both the arrival and departure processes implies that service times are directly observable. Indeed, this follows by rearranging the Lindley recursion,
\begin{align*}
H_i = D_i - \max\{D_{i-1} - G_i, 0\}, \qquad i = 1,2,\ldots.
\end{align*}
Consequently, the parameter vector ${\boldsymbol \theta}_\lambda \equiv \bigl(\E[H], \E[H^2], f_H^*(\lambda), f_H'(\lambda)\bigr)$ can be estimated directly via sample averages: if the available sample is $(H_1,H_2,\ldots,H_{n^{\step{k}}})$, then the estimator is
\begin{align*}
\left(
\frac{1}{{n^{\step{k}}}}\sum_{i=1}^{n^{\step{k}}} H_i,\;
\frac{1}{{n^{\step{k}}}}\sum_{i=1}^{n^{\step{k}}} H_i^2,\;
\frac{1}{{n^{\step{k}}}}\sum_{i=1}^{n^{\step{k}}} e^{-\lambda H_i},\;
-\frac{1}{{n^{\step{k}}}}\sum_{i=1}^{n^{\step{k}}} H_i e^{-\lambda H_i}
\right).
\end{align*}
In order to apply the theoretical convergence guarantees of Section~\ref{sec:alg}, we consider a truncated version of the estimation step. It is assumed that we have an upper bound $K$ on $\E[H^2]$. For some (small) $\epsilon>0$, we therefore let
\begin{align*}
\Theta_\epsilon
:=
\left[0,\frac{1}{\lambda_{\max}}-\epsilon\right]
\times [0,K-\epsilon]
\times [\epsilon,1]
\times [-\bar K,0],
\end{align*}
where $\bar{K}:=\sup_{x\ge 0} x e^{-\lambda_{\min} x} < \infty$; note that, for any $\lambda \in [\lambda_{\min},\lambda_{\max}]$,
\begin{align*}
\mathbb{E}[H e^{-\lambda H}]
\le
\mathbb{E}[H e^{-\lambda_{\min} H}]
\le
\sup_{x \ge 0} x e^{-\lambda_{\min} x}
=
\bar{K}
< \infty.
\end{align*}
    The parameter estimator is defined as 
\begin{align*}
\hat{\boldsymbol \theta}^{\step{k}}
:=
\bigl(\hat{h}_{1,n^{\step{k}}},\hat{h}_{2,n^{\step{k}}},\hat{h}_{3,n^{\step{k}}}(\lambda),\hat{h}_{4,n^{\step{k}}}(\lambda)\bigr),
\end{align*}
with components given by\begin{align*}
\hat{h}_{1,n^{\step{k}}}
&:= \max\!\left\{0,\min\!\left\{
\frac{1}{n^{\step{k}}}\sum_{i=1}^{n^{\step{k}}} H_i,
\frac{1}{\lambda_{\max}}-\epsilon
\right\}\right\},\\
\hat{h}_{2,n^{\step{k}}}
&:= \min\!\left\{
\frac{1}{n^{\step{k}}}\sum_{i=1}^{n^{\step{k}}} H_i^2,
K-\epsilon
\right\},\\
\hat{h}_{3,n^{\step{k}}}(\lambda)
&:= \max\!\left\{\epsilon,\min\!\left\{
\frac{1}{n^{\step{k}}}\sum_{i=1}^{n^{\step{k}}} e^{-\lambda H_i},
1
\right\}\right\},\\
\hat{h}_{4,n^{\step{k}}}(\lambda)
&:= -\frac{1}{n^{\step{k}}}\sum_{i=1}^{n^{\step{k}}} H_i e^{-\lambda H_i}.
\end{align*}
Note that the truncation level $\epsilon$ can be chosen arbitrarily small; in particular, one may take $\epsilon \ll \E[e^{-\lambda_{\max} H}]$, so that the truncation has no practical impact on the implementation of the algorithm.

{The following two lemmas verify the requirements of Assumption~\ref{ass:regularity} and will enable direct application of Theorem~\ref{thm:SGD_convergence_2} to this example. Lemmas \ref{lem:g_lip} and \ref{lem:theta_FD} focus on  Assumption~\ref{ass:regularity}(iv) (i.e., inequality \eqref{eq:Lips}) and Assumption~\ref{ass:regularity}(v) (i.e., the inequality \eqref{eq:theta_mixing}), respectively. The proofs are given in Appendix \ref{appendix:regularity}.
\begin{lemma}\label{lem:g_lip}
   With $\Theta=\Theta_\epsilon$, $g(\lambda;{\boldsymbol \theta}_\lambda)$ fulfills \eqref{eq:Lips}. 
\end{lemma}

\begin{lemma}\label{lem:theta_FD}
    Suppose that 
    $\E[H^4]<\infty$. Then \eqref{eq:theta_mixing} applies for any
    \begin{align}
    \label{eq:epsbound}
        \epsilon<\frac{1}{2}\min\left\lbrace\ \frac{1}{\lambda_{\max}}-\E[H]\,, {K-\E[H^2]\,,} \,f_H^*(\lambda_{\max})\right\rbrace.
    \end{align}
\end{lemma}

\begin{proposition}\label{prop:FD_convergence}
Suppose that $\lambda_{\min} < \lambda_* < \lambda_{\max}$,
$\E[H^4]<\infty$, 
$\lambda_{\max} \E[H] < 1$, and $\E[H^2] < K$ for some known $K < \infty$.
Then,
\[
\lambda^{\step{k}} \xrightarrow{\text{\rm a.s.}} \lambda_* \quad \text{as } k \to \infty.
\]
\end{proposition}

\begin{proof}
$\lambda_{\max}\E[H]<1$ ensures stability of the system for all $\lambda \in [\lambda_{\min},\lambda_{\max}]$. Consequently, the mean AoI is well defined and continuously differentiable on $\mathcal{L}$, which establishes Assumption~\ref{ass:regularity}(i).
Assumption~\ref{ass:regularity}(ii) follows from the explicit expression \eqref{eq:psi_grad-0} of the gradient.
Theorem~\ref{thm:FD_convex} implies that the objective is strictly
convex on $\mathcal{L}$, and hence has a unique minimizer.
Since $\lambda_{\min}<\lambda_*<\lambda_{\max}$, this minimizer is
interior and therefore satisfies
\[
g(\lambda_*;\bm{\theta}_{\lambda_*})+\kappa'(\lambda_*)=0.
\]
 Thus Assumption~\ref{ass:regularity}(iii) holds. Lemma~\ref{lem:g_lip} ensures that $g$ is uniformly Lipschitz continuous on $\Theta_\epsilon$, verifying Assumption~\ref{ass:regularity}(iv). Finally, noting that 
 $\lambda_{\max} \E[H] < 1$ and  $\E[H^2] < K < \infty$ ensures 
 the existence of $\epsilon >0$ satisfying \eqref{eq:epsbound}, 
 Lemma~\ref{lem:theta_FD} establishes, under $\E[H^4]<\infty$, the mean squared error bound required in Assumption~\ref{ass:regularity}(v). 
 Since all parts of Assumption~\ref{ass:regularity} are satisfied, the result follows directly from Theorem~\ref{thm:SGD_convergence_2}.
 \end{proof}
}


\subsection{Example 2: FCFS with background sources of unknown generation rate}\label{sec:application_B}

We now consider the FCFS M/GI/1 queue introduced in Example~2 of Section~\ref{ssec:examples}, where packets from a background stream arrive according to a Poisson process with unknown rate~$\lambda_{\bg}$. We assume throughout that the service-time distribution is known. Recall from \eqref{eq:EA_bg} that the mean AoI, denoted by $\E_{\lambda}[A]$, admits an explicit expression in this setting. We begin by establishing that $\E_{\lambda}[A]$ is a convex function of the generation rate~$\lambda$; the proof is deferred to Appendix~\ref{appendix:FB_convex}.
\begin{theorem}\label{thm:FB_convex}
The mean AoI $\E_{\lambda}[A]$ given by \eqref{eq:EA_bg} is a strictly convex function of $\lambda \in \left(0,1/\E[H]\right)$.
\end{theorem}

The gradient in this case is given by
\begin{align*}
g(\lambda;{\boldsymbol \theta}_\lambda)
&=
\frac{\E[H^{2}]}{2\left(1-(\lambda+\lambda_{\bg})\E[H]\right)^{2}}
-\frac{\lambda_{\bg}\E[H]}{\lambda^{2}}
-\frac{\E[H]}{\lambda\,\zeta(\lambda,\lambda_\bg)}
\\
&\qquad\qquad{}
-\frac{1-(\lambda+\lambda_{\bg})\E[H]}{\lambda^2 \zeta(\lambda,\lambda_\bg)^{2}}\,
\left(\zeta(\lambda,\lambda_\bg)+\lambda \zeta'(\lambda,\lambda_\bg)\right),
\end{align*}
where $\zeta(\lambda,\lambda_\bg)$ is the unique solution to \eqref{eq:Ex2FP}. 
We obtain by differentiating \eqref{eq:Ex2FP} that 
\[\zeta'(\lambda,\lambda_\bg) = f_H'(\lambda+\lambda_\bg -\lambda_\bg\zeta(\lambda,\lambda_\bg))\cdot(1-{\lambda_\bg\zeta'(\lambda,\lambda_\bg)}),\]
so that
\begin{equation}\label{eq:zeta_lambda_prime}
\zeta'(\lambda,\lambda_\bg) =\frac{f_H'(\lambda + \lambda_{\bg} - \lambda_{\bg} \zeta(\lambda,\lambda_\bg))}{1+\lambda_\bg{f_H'(\lambda + \lambda_{\bg} - \lambda_{\bg} \zeta(\lambda,\lambda_\bg))}}.
\end{equation}
Because  the service-time distribution is assumed to be known, the parameter vector has become
\begin{align*}
{\boldsymbol \theta}_\lambda\equiv \left(\lambda_{\bg},\zeta(\lambda,\lambda_\bg),  \zeta'(\lambda,\lambda_\bg)\right).
\end{align*}
However, there is effectively only one parameter to be estimated, because $\zeta(\lambda,\lambda_\bg)$ and $\zeta'(\lambda,\lambda_\bg)$ are computed by solving \eqref{eq:Ex2FP} and \eqref{eq:zeta_lambda_prime} given $(\lambda,\lambda_\bg,f_H^*)$. 
This means that we now have that the vector ${\boldsymbol \theta}_\lambda$ consists of just $\lambda_\bg$, and we continue by working with 
\[g^\circ(\lambda;\lambda_\bg):= g\big(\lambda; \,\lambda_{\bg},\zeta(\lambda,\lambda_\bg),  \zeta'(\lambda,\lambda_\bg)\big).\]
The parameter space, for the parameter $\lambda_\bg$, is then
\begin{align*}
\Theta_\epsilon:=\left[0,\frac{1}{\E[H]}-\epsilon-\lambda_{\max}\right].
\end{align*}

Inferring the background rate is closely related to the inverse problem of
estimating cross-traffic parameters from probe data \cite{BKV2009}. Indeed, they suggest estimating the arrival rate using stationary identities. However, their methodology cannot be directly applied because it relies on exact waiting time data, while we only observe total sojourn time. Moreover, the analysis in
\cite{BKV2009} is either noiseless or establishes consistency under an
i.i.d.\ idealization of probe delays, whereas Lemma~\ref{lem:theta_BG} provides a non-asymptotic $O(1/n)$ mean squared error bound under the dependence induced by the queue, uniformly over the generation rate and the initial state, as required by Assumption~\ref{ass:regularity}(v).

As in Example~1, all packets are informative. Since the packet-size
distribution is the same for both streams, the stationary delay
distribution is that of an M/GI/1 queue with total arrival rate
$\lambda+\lambda_\bg$. In particular,
\begin{equation}
1 - (\lambda+\lambda_\bg)\E[H]
=
\E\!\left[e^{-(\lambda+\lambda_\bg)D}\right].
\end{equation}
By the {\sc pasta} property, a tagged packet generated by the Poisson
source observes the stationary workload distribution upon arrival.
Hence, its stationary delay has the same distribution as the
stationary delay in the M/GI/1 queue with total arrival rate
$\lambda+\lambda_\bg$, even though background delays are not observed
and the full Lindley recursion cannot be reconstructed from the data.
This identity therefore provides an estimating equation for
$\lambda_\bg$. An estimator $\hat{\lambda}_\bg$ can then be obtained
from the sample of delays $(D_1,D_2,\ldots,D_n)$ by solving (with
respect to $\lambda_\bg$)
\begin{align*}
1 - (\lambda+\lambda_\bg)\E[H]
= \frac{1}{n}\sum_{i=1}^{n} e^{-(\lambda+\lambda_\bg)D_i},
\qquad
0 \leq \lambda_\bg < \frac{1}{\E[H]} - \lambda,
\end{align*}
which can equivalently be written as
\begin{equation}
\Psi_n(\lambda_\bg):=\frac{1}{n}\sum_{i=1}^{n} \psi(D_i; \lambda_\bg)=0,
\qquad
0 \leq \lambda_\bg < \frac{1}{\E[H]} - \lambda,
\label{eq:lambda_bg_estimator}
\end{equation}
where
\begin{equation}
\psi(\delta; \lambda_\bg) := 1 - (\lambda + \lambda_\bg)\E[H] - e^{-(\lambda + \lambda_\bg)\delta}.
\label{eq:Psi-d-l}
\end{equation}
Note that, for a fixed $\delta \geq 0$, the function $\psi(\delta;\lambda_\bg)$ is well-defined for all $\lambda_\bg \in \mathbb{R}$. Moreover, since $\psi(\delta;\lambda_\bg)$ is concave in $\lambda_\bg$, the empirical average
\[
    \frac{1}{n^{\step{k}}}\sum_{i=1}^{n^{\step{k}}} \psi(D_i;\lambda_\bg)
\]
is also concave in $\lambda_\bg \in \mathbb{R}$ for any fixed sample $(D_1,D_2,\ldots,D_{n^{\step{k}}})$. In addition, by \eqref{eq:Psi-d-l},
\begin{align}
\psi(\delta; -\lambda) &= 0, \qquad \delta \geq 0, 
\label{eq:solution-neg-lmd}
\\
\psi\Bigl(\delta; \frac{1}{\E[H]} - \lambda\Bigr) &< 0, \qquad \delta \geq 0.
\label{eq:neg-at-right-boundary-lmd}
\end{align}
Since a concave function that is not affine can have at most two zeros, \eqref{eq:solution-neg-lmd} implies that \eqref{eq:lambda_bg_estimator} admits at most one solution. Furthermore, \eqref{eq:neg-at-right-boundary-lmd} yields
\[
    \frac{1}{n^{\step{k}}}\sum_{i=1}^{n^{\step{k}}} \psi\Bigl(D_i; \frac{1}{\E[H]} - \lambda\Bigr) < 0,
\]
so that if
\[
    \frac{1}{n^{\step{k}}}\sum_{i=1}^{n^{\step{k}}} \psi(D_i; 0) \geq 0,
\]
then \eqref{eq:lambda_bg_estimator} admits a unique solution, while otherwise no solution exists.
Since $\psi(\delta;0)$ is increasing in $\delta$, this implies that
\eqref{eq:lambda_bg_estimator} may have no solution when the observed
sample $(D_1,\ldots,D_{n^{\step{k}}})$ takes relatively small values due
to random fluctuations. In this case, we set
$\hat{\lambda}_{\bg}=0$. Therefore, writing $\lambda_\bg^*$ for the
(unique, when it exists) solution of \eqref{eq:lambda_bg_estimator},
we obtain
\begin{equation}\label{eq:lambda_bg_hat}
\hat{\lambda}_{\bg}^{\step{k}}
=
\lambda_\bg^*\,{\boldsymbol 1}
\left\{
\frac{1}{n^{\step{k}}}
\sum_{i=1}^{n^{\step{k}}}\psi(D_i;0)\geq 0
\right\}.
\end{equation}

In practice, the estimator may be inefficient when the system is
temporarily heavily loaded, since a long observation period may be
required before the delay process becomes representative of its
stationary regime. To mitigate this transient effect, we augment the
algorithm with a recovery step: if the last observed delay in an
iteration exceeds a threshold $w$, packets are subsequently generated
at rate $\lambda_{\min}$ until an observation with delay at most $w$
is obtained. This recovery step does not alter the SGD update itself;
it is used only to control the transient workload between iterations.
The resulting implementation is summarized in
Algorithm~\ref{alg:background}.

\begin{algorithm}[t]
\caption{Online optimization for Example~2}
\label{alg:background}
\begin{algorithmic}[1]
\STATE Choose $\lambda^{\step{0}}\in\mathcal L$, $\epsilon>0$, and $w>0$.
\FOR{$k=0,1,2,\ldots$}
\STATE Generate $n^{\step{k}}$ packets at rate $\lambda^{\step{k}}$.
\STATE Compute estimator $\hat\lambda_{\bg}^{\step{k}}$ via \eqref{eq:lambda_bg_hat}.
\STATE Truncate estimator to $\Theta_\epsilon$: set
$\hat\lambda_{\bg}^{\step{k}}
:=\min\!\left\{\hat\lambda_{\bg}^{\step{k}},
{1/\E[H]}\,-\,\epsilon\,-\,\lambda_{\max}\right\}$.
\STATE Update
$\lambda^{\step{k+1}}
:=\Pi_{\mathcal L}\!\left\{\lambda^{\step{k}}
-\gamma^{\step{k}}\!\left[
g^\circ\!\left(\lambda^{\step{k}};\hat\lambda_{\bg}^{\step{k}}\right)
+\kappa'\!\left(\lambda^{\step{k}}\right)\right]\right\}$.
\IF{$D_{n^{\step{k}}}>w$}
\STATE Generate tagged packets at rate $\lambda_{\min}$ until $D_j\leq w$.
\ENDIF
\ENDFOR
\end{algorithmic}
\end{algorithm}

The following two lemmas verify the conditions in Assumption~\ref{ass:regularity}(iv)--(v) are satisfied by Algorithm~\ref{alg:background}. The proofs are given in Appendix \ref{appendix:regularity}.
\begin{lemma}\label{lem:g_lip_bg}
With $\Theta=\Theta_\epsilon$, $g^\circ(\lambda;\,\lambda_\bg)$ fulfills \eqref{eq:Lips}. 
\end{lemma}

\begin{lemma}\label{lem:theta_BG}
Suppose that there exists an $\epsilon>0$ such that $\lambda_{\bg}+\lambda_{\max}<1/\E[H]-\epsilon$ and that $\E[H^2]<\infty$. Then \eqref{eq:theta_mixing} applies.
\end{lemma}

Applying the same arguments as in the proof of Proposition~\ref{prop:FD_convergence}, with Lemmas~\ref{lem:g_lip_bg} and~\ref{lem:theta_BG} in place of Lemmas~\ref{lem:g_lip} and~\ref{lem:theta_FD}, respectively, yields the following result.
\begin{proposition}\label{prop:BG_convergence}
Suppose that $\lambda_{\bg}+\lambda_{\max}<1/\E[H]-\epsilon$ and that $\E[H^2]<\infty$.  Then,
\[
    \lambda^{\step{k}} \xrightarrow{\text{\rm a.s.}} \lambda_* \quad \text{as } k \to \infty.
\]
\end{proposition}

\subsection{Example 3: Preemptive LCFS System with unknown service-time distribution}\label{sec:application_C}

Recall that the gradient is given by, with the parameter vector ${\boldsymbol \theta}_\lambda 
\equiv 
\left(f_H^*(\lambda), f_H'(\lambda)\right)$,
\begin{align*}
g(\lambda; {\boldsymbol \theta}_\lambda)
&=
-\frac{f_H^*(\lambda) + \lambda f_H'(\lambda)}{\lambda^2 \{f_H^*(\lambda)\}^2}.
\end{align*}
The packet sizes are not always observed, but whenever a message is received then the packet size is simply the time since the last generation. This means that for every generated message we observe the indicator $\mathbf{1}\{H_i<G_i^\dag\}$. This yields the following unbiased estimator:
\begin{align}
\hat{\boldsymbol \theta}^{\step{k}}=\left(\frac{1}{n^{\step{k}}}\sum_{i=1}^{n^{\step{k}}}\mathbf{1}\{H_i<G_i^\dag\}, \ -\frac{1}{n^{\step{k}}}\sum_{i=1}^{n^{\step{k}}}H_i\mathbf{1}\{H_i<G_i^\dag\}\right)\ .
\end{align}
Note that the second coordinate can be extracted from the data even though $H_i$ is not observed if $H_i>G_i^\dag$, in which case the indicator returns a zero.

As both parameters are constructed from i.i.d.\ observations, the
conditions in Assumption~\ref{ass:regularity}(iv)--(v) can be verified
using the same type of arguments as in Example~1. In particular, the
first component of $\hat{\boldsymbol\theta}^{\step{k}}$ is bounded, while
the second component is bounded in absolute value by $H_i$. Thus,
Assumption~\ref{ass:regularity}(v) holds under
$\E[H^2]<\infty$. Moreover, the Lipschitz condition in
Assumption~\ref{ass:regularity}(iv) follows directly from the form of
the gradient, since
\[
    g(\lambda;\boldsymbol\theta)
    =
    -\frac{\theta_1+\lambda\theta_2}{\lambda^2\theta_1^2}
\]
is continuously differentiable with respect to $\boldsymbol\theta$
on the relevant parameter space, provided $\theta_1$ is bounded away
from zero. Consequently, under the stated moment and parameter-space
conditions, the convergence result of
Theorem~\ref{thm:SGD_convergence_2} applies.

\subsection{Example 4: Generate-at-will system with threshold waiting}\label{sec:application_D}

The dispatcher observes packet arrivals and sends a new packet if more than $1/\lambda$ units have passed since the last generation, and otherwise waits for the remainder; this means that, by \eqref{eq:xi},
\begin{align*}
\xi_\lambda(H)=\mathbf{1}\left\lbrace H<\frac{1}{\lambda}\right\rbrace\left(\frac{1}{\lambda}-H\right).
\end{align*}
We assume $H$ has an absolutely continuous distribution. In this case the gradient is given by \eqref{eq:gEx4},
where 
\[
    G_n^{\dag} = H_{n-1} + \xi_{\lambda}(H_{n-1}),
    \qquad
    D_n^{\dag} = H_n.
\]
The parameter vector \eqref{eq:thetaEx4} can be estimated directly from the empirical means:
\begin{align*}
\hat{\boldsymbol \theta}^{\step{k}}=\left(\frac{1}{n^{\step{k}}}\sum_{i=1}^{n^{\step{k}}}(H_i+\xi_{\lambda^{\step{k}}}(H_{i})),\ \frac{1}{n^{\step{k}}}\sum_{i=1}^{n^{\step{k}}}(H_i+\xi_{\lambda^{\step{k}}}(H_{i}))^2, \ \frac{1}{n^{\step{k}}}\sum_{i=1}^{n^{\step{k}}}\mathbf{1}\left\lbrace H_i<\frac{1}{\lambda^{\step{k}}}\right\rbrace\right)\ .
\end{align*}
The regularity conditions required for convergence can be verified
using the same approach as in the preceding examples. In particular,
the components of $\hat{\boldsymbol\theta}^{\step{k}}$ are empirical
averages of functions of the i.i.d.\ service times. Under the
corresponding finite-moment conditions, their mean squared errors are
of order $O((n^{\step{k}})^{-1})$, yielding
Assumption~\ref{ass:regularity}(v). Moreover, on the compact parameter
space corresponding to $\lambda\in\mathcal L$, the gradient in
\eqref{eq:gEx4} is continuously differentiable with respect to the
components of $\boldsymbol\theta$. Hence it is Lipschitz continuous on
this parameter space, which gives
Assumption~\ref{ass:regularity}(iv). Consequently, provided the
objective has a unique minimizer in $\mathcal L$, the convergence
result of Theorem~\ref{thm:SGD_convergence_2} applies.


\section{Numerical Study}\label{sec:num}

This section evaluates the proposed structure-aware online optimization
framework on Examples~1--4. We first describe the common experimental setup
and performance measures, and then give an overview of the algorithm's
behavior across the four examples. We next examine several features in
greater detail, including the effect of service-time variability, stability
and its detection under incomplete information, and statistical uncertainty
and stopping. Finally, we compare the proposed structure-aware method with a
model-free Kiefer--Wolfowitz benchmark to quantify the value of exploiting
the analytical gradient representation.

\subsection{Experimental setup}
\label{sec:sim_setup}

Denote the optimal value of the objective function by
\begin{align*}
Q_* :=
\min_{\lambda \in \mathcal{L}}
\bigl(\E_{\lambda}[A]+\kappa(\lambda)\bigr)
=
\E_{\lambda_*}[A]+\kappa(\lambda_*).
\end{align*}
Throughout, we assess performance using four quantities. First, we track the
trajectory of $\lambda^{\step{k}}$ and its mean squared error
$\E[(\lambda^{\step{k}}-\lambda_*)^2]$. Second, we consider the objective
gap
$
\E_{\lambda^{\step{k}}}[A]
+
\kappa(\lambda^{\step{k}})
-
Q_*,
$
which is non-negative and vanishes at an optimum. Third, we report the
realized time-average cost
\[
    \widehat{Q}_\tau
    :=
    \frac{1}{\tau}
    \sum_{k=0}^{M_{\tau}-1}
    \left\{
    \overline{A}^{\step{k}}
    +
    (t^{\step{k+1}}-t^{\step{k}})
    \kappa(\lambda^{\step{k}})
    \right\},
\]
obtained from the sample path, cf.\ \eqref{eq:objective}. Finally, we report
the consumed budget, measured either by the cumulative number of informative
packets
\[
    N^{\step{k}}
    =
    \sum_{j=1}^k n^{\step{j}},
\]
or by elapsed system time. We report performance against budget rather than
iteration count, since the methods may consume different numbers of samples
per iteration.

Table~\ref{tab:H_families} summarizes the service-time distributions used
throughout the numerical study. All distributions are normalized to
$\E[H]=1$, so that the stability condition in Examples~1 and~2 becomes
$\lambda_{\mathrm{tot}}<1$, where $\lambda_{\mathrm{tot}}$ denotes the total
arrival rate. Within each distribution family, the {\it squared coefficient
                                                    of variation}
\[
    c^2
    :=
    \frac{\Var[H]}{\E[H]^2}
\]
determines the relevant shape parameter.

\begin{table}
\centering
\footnotesize
\begin{tabular}{llccl}
\hline
\rule{0pt}{2.4ex}
Label & Distribution & Parameters & $c^2$ & $f_H^*(s)$ \\[0.3ex]
\hline
\rule{0pt}{2.6ex}
$\mathrm{D}$
& Deterministic
& $H\equiv 1$
& $0$
& $e^{-s}$ \\[0.6ex]
$\mathrm{E}_4$
& Erlang-$4$
& shape $4$, rate $4$
& $1/4$
& $\left(\dfrac{4}{4+s}\right)^{\!4}$ \\[1.4ex]
$\mathrm{M}$
& Exponential
& rate $1$
& $1$
& $\dfrac{1}{1+s}$ \\[1.4ex]
$\mathrm{H}_2^{(4)}$
& Hyperexponential
& $p_{1,2}=\frac{1}{2}\bigl(1\pm\sqrt{3/5}\bigr)$,
$\mu_i=2p_i$
& $4$
& $\displaystyle\sum_{i=1,2}\frac{p_i\mu_i}{\mu_i+s}$ \\[1.6ex]
$\mathrm{H}_2^{(9)}$
& Hyperexponential
& $p_{1,2}=\frac{1}{2}\bigl(1\pm\sqrt{4/5}\bigr)$,
$\mu_i=2p_i$
& $9$
& $\displaystyle\sum_{i=1,2}\frac{p_i\mu_i}{\mu_i+s}$ \\[1.6ex]
$\mathrm{LN}$
& Lognormal
& $\sigma^2=\ln 5$, $\mu=-\sigma^2/2$
& $4$
& no closed form \\[0.6ex]
$\mathrm{P}_{4.5}$
& Pareto
& $\alpha=4.5$, $x_m=7/9$
& $4/45$
& $\alpha (x_m s)^{\alpha}\,\Gamma(-\alpha,x_m s)$ \\[0.6ex]
$\mathrm{P}_{2.5}$
& Pareto
& $\alpha=2.5$, $x_m=3/5$
& $4/5$
& $\alpha (x_m s)^{\alpha}\,\Gamma(-\alpha,x_m s)$ \\[0.6ex]
\hline
\end{tabular}
\caption{Service-time families used in the numerical study, normalized to
         $\E[H]=1$. For the Pareto family,
         $x_m=1-1/\alpha$ and $c^2=1/\{\alpha(\alpha-2)\}$;
         $\Gamma(\cdot,\cdot)$ denotes the upper incomplete gamma function.}
\label{tab:H_families}
\end{table}

For the proposed method, unless otherwise stated, we use the projection
interval $\mathcal{L}=[0.3,0.8]$, step size
$\gamma^{\step{k}}=c_1k^{-1}$. The step-size constant is scaled by the local curvature of the objective,
\begin{equation}
c_1
=
\frac{2}{\mu},
\qquad
\mu
:=
g'(\lambda_*;\bm{\theta}_{\lambda_*})
+
\kappa''(\lambda_*),
\label{eq:curvature-def}
\end{equation}
so that the examples are placed on a comparable optimization scale. This
normalization is used only in the numerical experiments; the convergence
result in Theorem~\ref{thm:SGD_convergence_2} holds for any admissible choice
of the step-size constant. Since $\mu$ depends on the unknown distributional
parameters at the optimum, it is not available to the algorithm and is used
here only for calibration.

As the generation cost function is known, it does not play a role in the performance of the learning algorithm in Examples 1, 2, and 4. We set $\kappa(\lambda)=0$ for all $\lambda>0$ for these examples. In Example 3 the AoI is monotone decreasing with $\lambda$ (as long as the support of $H$ starts at zero), hence the optimization problem is degenerate without generation cost. The chosen cost function for Example 3 is set to $\kappa(\lambda)=0.25\,\lambda$.

The projection interval is chosen to contain the optima considered while
avoiding the steep gradient behavior near the stability boundaries. A wider
interval such as $\mathcal{L}=[0.05,0.95]$ produces a longer initial
transient and can substantially increase the realized time-average cost,
since low generation rates consume disproportionate amounts of elapsed
time. For Example~1 we set $\mathcal{L}=[0.3,0.8]$, and for Example~2, $\mathcal{L}=[0.15,0.65]$ keeps $\lambda_*$ interior while respecting the joint stability constraint. Examples~3 and 4 admit larger optimal rates and we therefore set $\mathcal{L}=[0.2,5]$ and $\mathcal{L}=[0.3,1.2]$, respectively.

In Example~2, the algorithm is Algorithm~\ref{alg:background}, including its stabilizing burn step, and the burn packets are charged to the budget: they contribute to both $N^{\step{k}}$ and the realized cost in exactly the same way as the informative packets of the iteration. We set $w=100$, that is, $w=100\,\E[H]$. Since Lemma~\ref{lem:D_mixing} holds for any finite threshold, $w$ is not a tuning parameter; choosing it sufficiently large makes the burn inactive in all iterations. Thus, the threshold is relevant for the theoretical analysis but does not affect the numerical performance.

\subsection{Convergence across the four examples}
\label{sec:sim_overview}

\begin{figure}
\centering
\begin{subfigure}[t]{0.48\linewidth}
\centering
\begin{tikzpicture}
\begin{axis}[
  width=6.2cm,
  height=4.4cm,
  xlabel={iteration $k$},
  ylabel={$\lambda^{\step{k}}/\lambda_*$},
  ylabel style={rotate=-90, anchor=east, xshift=-2pt, font=\small},
  xlabel style={font=\small},
  tick label style={font=\footnotesize},
  legend style={font=\scriptsize},
  grid=major,
  grid style={line width=0.25pt, draw=gray!22},
  axis lines=left,
  every axis plot/.append style={line width=0.7pt},
  xmin=0, xmax=100,
  ymin=0, ymax=2.6
]
  \addplot[red, very thick, dashed, no marks] coordinates {(0,1) (100,1)};
  \addplot[black, line width=0.6pt, no marks] table[x=k, y=rel] {xex_ex1.dat};
\end{axis}
\end{tikzpicture}
\caption{Example 1, FCFS}
\end{subfigure}\hfill
\begin{subfigure}[t]{0.48\linewidth}
\centering
\begin{tikzpicture}
\begin{axis}[
  width=6.2cm,
  height=4.4cm,
  xlabel={iteration $k$},
  ylabel={$\lambda^{\step{k}}/\lambda_*$},
  ylabel style={rotate=-90, anchor=east, xshift=-2pt, font=\small},
  xlabel style={font=\small},
  tick label style={font=\footnotesize},
  legend style={font=\scriptsize},
  grid=major,
  grid style={line width=0.25pt, draw=gray!22},
  axis lines=left,
  every axis plot/.append style={line width=0.7pt},
  xmin=0, xmax=100,
  ymin=0, ymax=2.6
]
  \addplot[red, very thick, dashed, no marks] coordinates {(0,1) (100,1)};
  \addplot[black, line width=0.6pt, no marks] table[x=k, y=rel] {xex_ex2.dat};
\end{axis}
\end{tikzpicture}
\caption{Example 2, background stream}
\end{subfigure}\\[1ex]
\begin{subfigure}[t]{0.48\linewidth}
\centering
\begin{tikzpicture}
\begin{axis}[
  width=6.2cm,
  height=4.4cm,
  xlabel={iteration $k$},
  ylabel={$\lambda^{\step{k}}/\lambda_*$},
  ylabel style={rotate=-90, anchor=east, xshift=-2pt, font=\small},
  xlabel style={font=\small},
  tick label style={font=\footnotesize},
  legend style={font=\scriptsize},
  grid=major,
  grid style={line width=0.25pt, draw=gray!22},
  axis lines=left,
  every axis plot/.append style={line width=0.7pt},
  xmin=0, xmax=100,
  ymin=0, ymax=2.6
]
  \addplot[red, very thick, dashed, no marks] coordinates {(0,1) (100,1)};
  \addplot[black, line width=0.6pt, no marks] table[x=k, y=rel] {xex_ex3.dat};
\end{axis}
\end{tikzpicture}
\caption{Example 3, preemptive LCFS}
\end{subfigure}\hfill
\begin{subfigure}[t]{0.48\linewidth}
\centering
\begin{tikzpicture}
\begin{axis}[
  width=6.2cm,
  height=4.4cm,
  xlabel={iteration $k$},
  ylabel={$\lambda^{\step{k}}/\lambda_*$},
  ylabel style={rotate=-90, anchor=east, xshift=-2pt, font=\small},
  xlabel style={font=\small},
  tick label style={font=\footnotesize},
  legend style={font=\scriptsize},
  grid=major,
  grid style={line width=0.25pt, draw=gray!22},
  axis lines=left,
  every axis plot/.append style={line width=0.7pt},
  xmin=0, xmax=100,
  ymin=0, ymax=2.6
]
  \addplot[red, very thick, dashed, no marks] coordinates {(0,1) (100,1)};
  \addplot[black, line width=0.6pt, no marks] table[x=k, y=rel] {xex_ex4.dat};
\end{axis}
\end{tikzpicture}
\caption{Example 4, generate-at-will}
\end{subfigure}
\caption{The estimated rate in all four examples, normalized by $\lambda_*$, on a common axis. All four panels use lognormal service with $\E[H]=1$ and $c^2=4$, so that they differ in the model rather than in the input; the Example~2 panel has $\lambda_{\mathrm{bg}}=0.3$ and the Example~3 panel $\kappa(\lambda)=0.25\,\lambda$. Each run uses $K=100$ iterations with $n^{\step{k}}=25k$. The four objectives differ substantially in curvature, yet the normalized trajectories have comparable shapes under the curvature scaling in \eqref{eq:curvature-def}.}
\label{fig:xexn}
\end{figure}
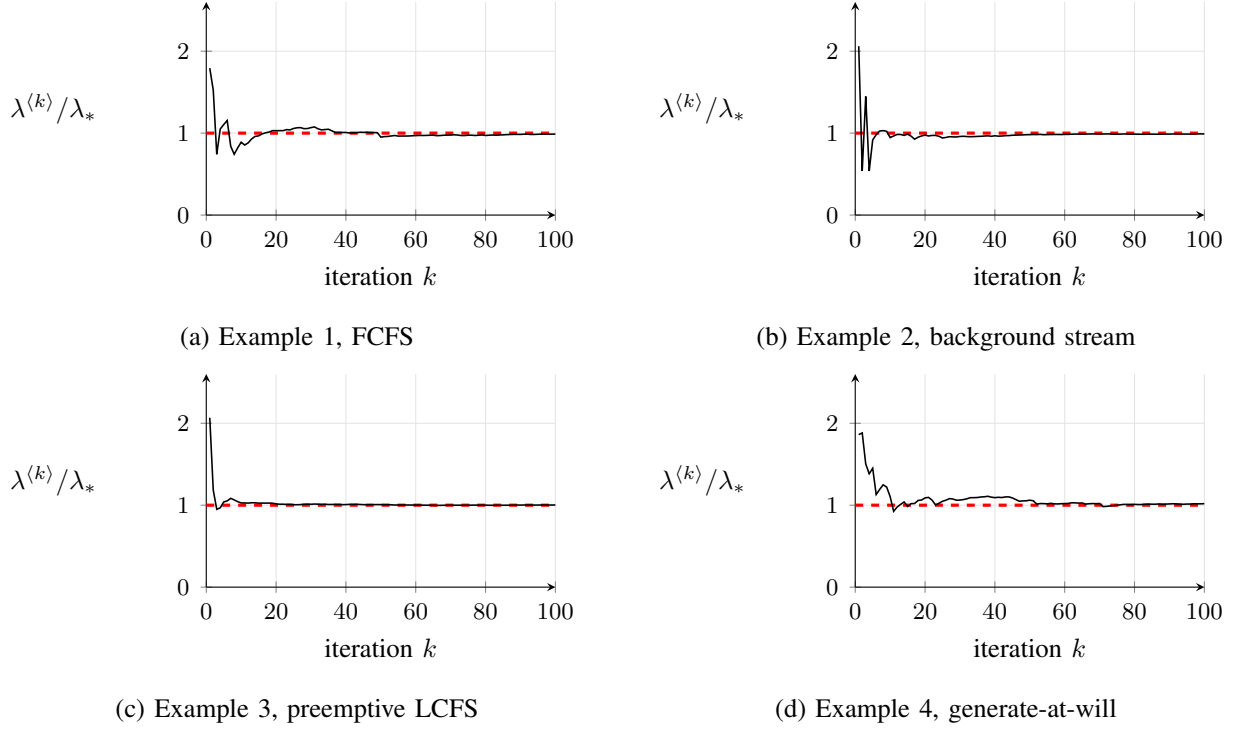

\begin{figure}
\centering
\begin{tikzpicture}
\begin{axis}[
  width=12.0cm,
  height=7cm,
  xlabel={iteration $k$},
  ylabel={},
  ylabel style={rotate=-90, anchor=east, xshift=-2pt},
  grid=major,
  grid style={line width=0.25pt, draw=gray!22},
  axis lines=left,
  every axis plot/.append style={line width=0.7pt},
  xmin=0, xmax=200,
  legend pos=north east,
  legend cell align=left
]
  \addplot[black, no marks] table[x=k, y=lambda] {ex2_traj.dat};
  \addlegendentry{$\lambda^{\step{k}}$}
  \addplot[red, very thick, dashed, no marks] coordinates {(0,0.3644) (200,0.3644)};
  \addlegendentry{$\lambda_*$}
  \addplot[blue, no marks] table[x=k, y=lbg_hat] {ex2_traj.dat};
  \addlegendentry{$\hat\lambda_{\mathrm{bg}}^{\step{k}}$}
  \addplot[blue, thick, dashed, no marks] coordinates {(0,0.300) (200,0.300)};
  \addlegendentry{$\lambda_{\mathrm{bg}}$}
\end{axis}
\end{tikzpicture}
\caption{Example 2: iterates of Algorithm~\ref{alg:background} and the
                         $Z$-estimator of the background rate, with $\E[H]=1$ and
                         $\lambda_{\mathrm{bg}}=0.3$, so the joint stability boundary is
                         $\lambda<0.7$. The monitor observes neither background arrivals nor their
                         delays. Exponential service, $\mathcal{L}=[0.15,0.65]$,
                         $\lambda^{\step{0}}=0.6$ and $n^{\step{k}}=100k$ over $200$ iterations; the
                         burn step was not triggered.}
\label{fig:ex2_traj}
\end{figure}
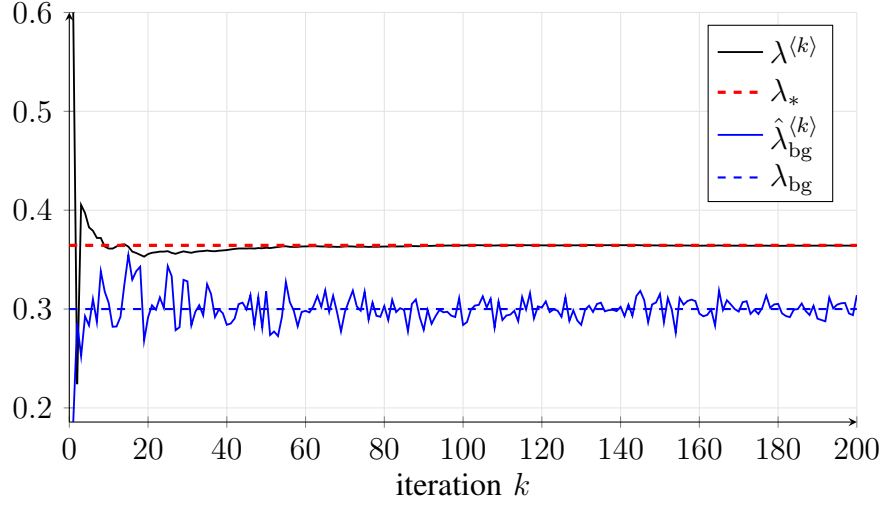

Figure~\ref{fig:xexn} gives an overview of the proposed algorithm across all
four examples using a common service-time distribution. The same tuning rule
is used throughout, with the starting point set to $85\%$ of the projection
interval. The terminal iterates are within $1.9\%$
of $\lambda_*$ in all four cases. The optimal rates range from approximately
$0.28$ in Example~2 to $2.07$ in Example~3, reflecting differences between
the underlying models rather than differences in the optimization procedure.
After an initial descent, the iterates settle near their respective optima.

Overall, these results show that the proposed structure-aware algorithm
successfully converges across all four examples, despite the substantial
differences in their queueing models, information structures, and objective
scales. In particular, the common tuning rule is sufficient to obtain
accurate solutions without example-specific adjustment of the optimization
procedure. The figure therefore provides a first indication that the
analytical gradient representation is not tied to a particular model
structure, but can support effective online optimization across the range of
incomplete-information settings considered here.

For Example~2, Figure~\ref{fig:ex2_traj} additionally shows the estimate $\hat\lambda_{\bg}^{\step{k}}$ of the unobserved background rate. Despite being inferred solely from the monitor's own observations, the estimate converges to the true value while the generation rate simultaneously
approaches $\lambda_*$. The feasible upper limit
$1/\E[H]-\lambda_{\bg}$ is therefore learned implicitly. The rate estimate
can remain noisy after the optimization has essentially converged,
indicating that accurate pointwise estimation of $\lambda_{\bg}$ is not
necessary for effective optimization.

\subsection{Effect of service-time variability}
\label{sec:sim_A}

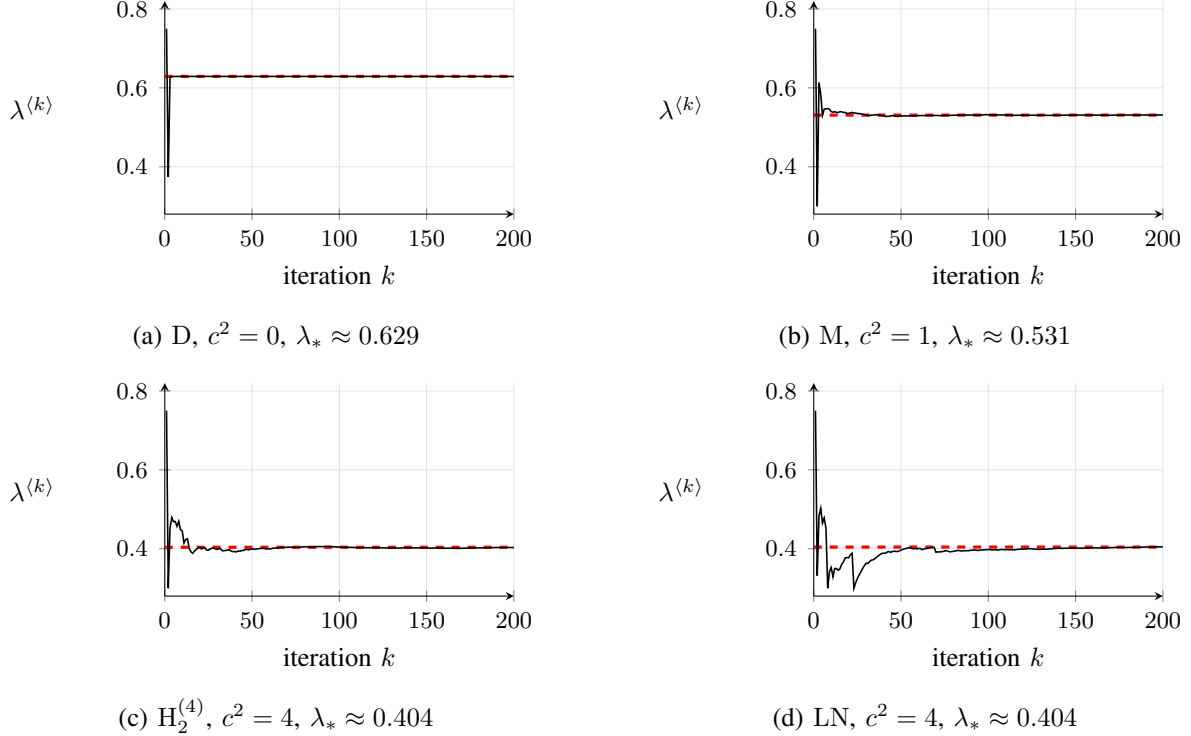
\begin{figure}
\centering
\begin{subfigure}[t]{0.48\linewidth}
\centering
\begin{tikzpicture}
\begin{axis}[
  width=6.2cm,
  height=4.4cm,
  xlabel={iteration $k$},
  ylabel={$\lambda^{\step{k}}$},
  ylabel style={rotate=-90, anchor=east, xshift=-2pt, font=\small},
  xlabel style={font=\small},
  tick label style={font=\footnotesize},
  legend style={font=\scriptsize},
  grid=major,
  grid style={line width=0.25pt, draw=gray!22},
  axis lines=left,
  every axis plot/.append style={line width=0.7pt},
  xmin=0, xmax=200,
  ymin=0.28, ymax=0.82
]
  \addplot[red, very thick, dashed, no marks] coordinates {(0,0.6291) (200,0.6291)};
  \addplot[black, line width=0.6pt, no marks] table[x=k, y=lambda] {ex1_traj_D.dat};
\end{axis}
\end{tikzpicture}
\caption{$\mathrm{D}$, $c^2=0$, $\lambda_*\approx0.629$}
\end{subfigure}\hfill
\begin{subfigure}[t]{0.48\linewidth}
\centering
\begin{tikzpicture}
\begin{axis}[
  width=6.2cm,
  height=4.4cm,
  xlabel={iteration $k$},
  ylabel={$\lambda^{\step{k}}$},
  ylabel style={rotate=-90, anchor=east, xshift=-2pt, font=\small},
  xlabel style={font=\small},
  tick label style={font=\footnotesize},
  legend style={font=\scriptsize},
  grid=major,
  grid style={line width=0.25pt, draw=gray!22},
  axis lines=left,
  every axis plot/.append style={line width=0.7pt},
  xmin=0, xmax=200,
  ymin=0.28, ymax=0.82
]
  \addplot[red, very thick, dashed, no marks] coordinates {(0,0.5310) (200,0.5310)};
  \addplot[black, line width=0.6pt, no marks] table[x=k, y=lambda] {ex1_traj_M.dat};
\end{axis}
\end{tikzpicture}
\caption{$\mathrm{M}$, $c^2=1$, $\lambda_*\approx0.531$}
\end{subfigure}\\[1ex]
\begin{subfigure}[t]{0.48\linewidth}
\centering
\begin{tikzpicture}
\begin{axis}[
  width=6.2cm,
  height=4.4cm,
  xlabel={iteration $k$},
  ylabel={$\lambda^{\step{k}}$},
  ylabel style={rotate=-90, anchor=east, xshift=-2pt, font=\small},
  xlabel style={font=\small},
  tick label style={font=\footnotesize},
  legend style={font=\scriptsize},
  grid=major,
  grid style={line width=0.25pt, draw=gray!22},
  axis lines=left,
  every axis plot/.append style={line width=0.7pt},
  xmin=0, xmax=200,
  ymin=0.28, ymax=0.82
]
  \addplot[red, very thick, dashed, no marks] coordinates {(0,0.4038) (200,0.4038)};
  \addplot[black, line width=0.6pt, no marks] table[x=k, y=lambda] {ex1_traj_H2_c4.dat};
\end{axis}
\end{tikzpicture}
\caption{$\mathrm{H}_2^{(4)}$, $c^2=4$, $\lambda_*\approx0.404$}
\end{subfigure}\hfill
\begin{subfigure}[t]{0.48\linewidth}
\centering
\begin{tikzpicture}
\begin{axis}[
  width=6.2cm,
  height=4.4cm,
  xlabel={iteration $k$},
  ylabel={$\lambda^{\step{k}}$},
  ylabel style={rotate=-90, anchor=east, xshift=-2pt, font=\small},
  xlabel style={font=\small},
  tick label style={font=\footnotesize},
  legend style={font=\scriptsize},
  grid=major,
  grid style={line width=0.25pt, draw=gray!22},
  axis lines=left,
  every axis plot/.append style={line width=0.7pt},
  xmin=0, xmax=200,
  ymin=0.28, ymax=0.82
]
  \addplot[red, very thick, dashed, no marks] coordinates {(0,0.4043) (200,0.4043)};
  \addplot[black, line width=0.6pt, no marks] table[x=k, y=lambda] {ex1_traj_LN.dat};
\end{axis}
\end{tikzpicture}
\caption{$\mathrm{LN}$, $c^2=4$, $\lambda_*\approx0.404$}
\end{subfigure}
\caption{Example 1: iterates of the proposed algorithm against the optimal rate $\lambda_*$ (dashed), for four of the service-time families of Table~\ref{tab:H_families} with $\E[H]=1$ and $\kappa(\lambda)=0$ for all $\lambda>0$.}
\label{fig:ex1_traj}
\end{figure}

We next examine how service-time variability affects the optimal operating
rate and the performance of the proposed algorithm. Example~1 provides a
natural setting for this study because the service-time distribution is
unknown, while the relevant service times can be recovered from the observed
timestamps. We therefore vary the service-time distribution in Example~1
while keeping the remaining setup fixed. The cross-example results in
Figure~\ref{fig:xexn} show that the resulting convergence behavior is not
specific to Example~1.

Figure~\ref{fig:ex1_traj} shows the convergence behavior for four
service-time families. Starting from $\lambda^{\step{0}}=0.75$, the iterates
reach the vicinity of $\lambda_*$ within a few iterations in all four cases.
Although the first update reaches the lower projection boundary, the
subsequent iterates quickly approach and remain near their respective
optima, with $\widehat{Q}_\tau$ agreeing with $Q_*$ to three decimal places.
Thus, the proposed method remains effective across different levels and
shapes of service-time variability.

The main effect of service-time variability is on the optimal operating
rate, which decreases from $\lambda_*\approx0.629$ for deterministic service
to $\lambda_*\approx0.531$ for exponential service and
$\lambda_*\approx0.404$ for the two distributions with $c^2=4$. Despite
their different shapes, the two $c^2=4$ distributions have nearly identical
optimal rates. More generally, across six service-time distributions with
$\E[H]=1$ and $c^2=4$, the optimal rate varies by less than $1\%$, whereas
the optimal objective value varies by about $7\%$. Thus, the first two
moments of $H$ provide a good indication of the optimal operating rate, but
are substantially less informative about the optimal objective value.

The deterministic case provides a useful benchmark: the service times are
recovered exactly from the timestamps, so $\hat{\bm{\theta}}$ is error-free
for every sample size and the proposed method reduces to exact projected
gradient descent. For the other service-time families, the observed
differences therefore reflect both service-time variability and the
statistical error in estimating the unknown distribution.

\subsection{Stability and incomplete-information safeguards}
\label{sec:sim_stability}

We next examine behavior near the stability boundary, focusing on the
consequences of violating the theoretical stability condition and on the
effectiveness of the data-driven safeguards when the stability region is
unknown. 
We implement two safeguards that replace the fixed truncation set
$\Theta_\epsilon$ with the data-driven cap $1/\lambda^{\step{k}}-\epsilon$.
The safeguards are activated when an iteration is flagged as approaching
the stability boundary. For Examples~1, 3, and~4, an iteration is flagged
when the estimated load reaches one to within a margin
$\epsilon_{\mathrm{stab}}$, that is, when
\begin{equation}
\lambda^{\step{k}}\,\hat h_1^{\step{k}}
\;\geq\;
1-\epsilon_{\mathrm{stab}},
\label{eq:stab-trigger}
\end{equation}
where $\hat h_1^{\step{k}}$ denotes the estimate of $\E[H]$ from the current
decision interval. In Example~2, where $\E[H]$ is known but
$\lambda_{\bg}$ is not, an iteration is instead flagged when the
$Z$-estimator \eqref{eq:lambda_bg_hat} saturates at the endpoint of its
parameter space. On a flagged iteration, the two safeguards act as follows.

\begin{itemize}
\item \textit{Reset.} The gradient step is discarded and the iterate is
returned to the smallest feasible rate,
$\lambda^{\step{k+1}}=\lambda_{\min}$.
\item \textit{Adaptive projection.} The gradient step is taken as usual, but
projected onto the contracted interval
$[\lambda_{\min},\,\min\{\lambda_{\max},(1-\epsilon_{\mathrm{stab}})/\hat
  h_1^{\step{k}}\}]$ instead of onto $\mathcal{L}$; in Example~2 the upper
endpoint is $1/\E[H]-\hat\lambda_{\bg}^{\step{k}}-\epsilon_{\mathrm{stab}}$.
If the contracted interval is empty it is collapsed to $\lambda_{\min}$.
\end{itemize}
Throughout we take $\epsilon=10^{-3}$ for the truncation margin and
$\epsilon_{\mathrm{stab}}=10^{-2}$ for the stability margin; the latter is the larger of the two because it guards a quantity that must be estimated along the run rather than a fixed modelling constant.

Proposition~\ref{prop:FD_convergence} assumes
$\E[H]<(1-\epsilon)/\lambda_{\max}$. When this condition fails, the fixed
truncation set $\Theta_\epsilon$ caps $\hat h_1$ below $\E[H]$, introducing
asymptotic bias. In the left panel of Figure~\ref{fig:stability}, we set
$\E[H]=1$ and $\lambda_{\max}=1.5$, so that
$\lambda_{\max}\E[H]=1.5>1$. The fixed-truncation algorithm consequently
converges to a perturbed stationary point, whereas the data-driven
truncation $1/\lambda^{\step{k}}-\epsilon$, combined with either the reset
or adaptive-projection safeguard, restores convergence to $\lambda_*$. Both
safeguards are triggered only a small number of times. These results are
illustrative; we do not claim a convergence proof for the safeguarded
variants.

Example~2 considers the same issue when the stability boundary depends on
the unknown background rate. Proposition~\ref{prop:BG_convergence} requires
$\lambda_{\bg}+\lambda_{\max}<1/\E[H]-\epsilon$, which cannot be verified
a priori because $\lambda_{\bg}$ is unobserved. The weaker condition
$\lambda_{\bg}+\lambda_{\min}<1/\E[H]-\epsilon$ only ensures that a stable
operating rate exists. 
The $Z$-estimator nevertheless provides an intrinsic instability diagnostic: when $\lambda+\lambda_{\bg}>1/\E[H]$, the system is unstable and the observed delays grow without a stationary distribution. Consequently, the empirical transform tends toward zero over increasingly long observation periods, causing the estimated root to approach the boundary of its parameter space.
As shown in the right panel of
Figure~\ref{fig:stability}, saturation is largely absent in the stable
region but becomes pronounced beyond the stability boundary. Thus, the
estimator can signal instability even when $\lambda_{\bg}$ is unobserved.
The choice of $\epsilon$ involves a tradeoff: a larger margin detects
instability earlier but may also trigger unnecessarily close to the
stability boundary.

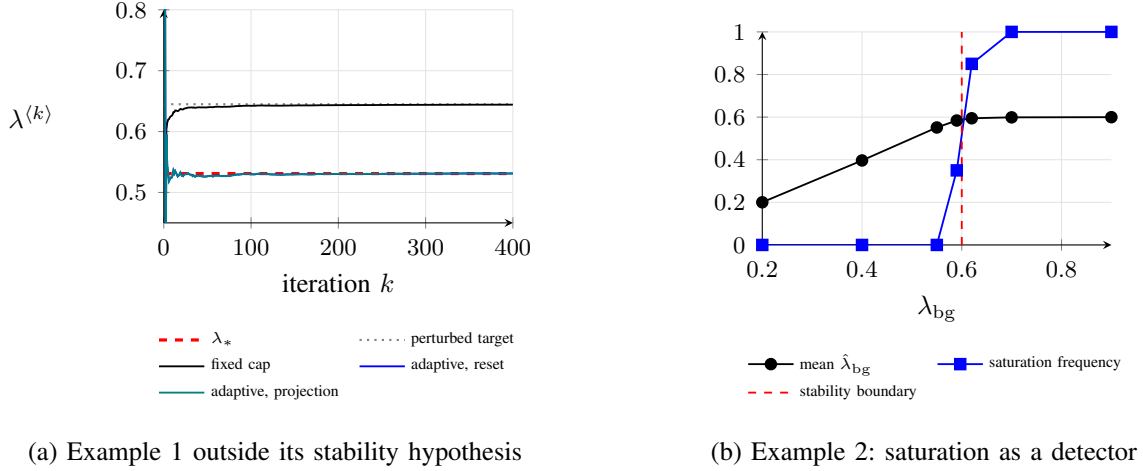
\begin{figure}
\centering
\begin{subfigure}[t]{0.48\linewidth}
\centering
\begin{tikzpicture}
\begin{axis}[
  width=6.2cm,
  height=4.4cm,
  xlabel={iteration $k$},
  ylabel={$\lambda^{\step{k}}$},
  ylabel style={rotate=-90, anchor=east, xshift=-2pt, font=\small},
  xlabel style={font=\small},
  tick label style={font=\footnotesize},
  legend style={font=\scriptsize},
  grid=major,
  grid style={line width=0.25pt, draw=gray!22},
  axis lines=left,
  every axis plot/.append style={line width=0.7pt},
  xmin=0, xmax=400,
  ymin=0.45, ymax=0.80,
  legend cell align=left,
  legend style={at={(0.5,-0.46)}, anchor=north, legend columns=2, font=\tiny, draw=none, /tikz/every even column/.append style={column sep=1ex}}
]
  \addplot[red, very thick, dashed, no marks] coordinates {(0,0.5310) (400,0.5310)};
  \addlegendentry{$\lambda_*$}
  \addplot[gray, thick, dotted, no marks] coordinates {(0,0.6449) (400,0.6449)};
  \addlegendentry{perturbed target}
  \addplot[black, no marks] table[x=k, y=fixed_none] {ex1_unstable.dat};
  \addlegendentry{fixed cap}
  \addplot[blue, no marks] table[x=k, y=adaptive_reset] {ex1_unstable.dat};
  \addlegendentry{adaptive, reset}
  \addplot[teal, no marks] table[x=k, y=adaptive_proj] {ex1_unstable.dat};
  \addlegendentry{adaptive, projection}
\end{axis}
\end{tikzpicture}
\caption{Example 1 outside its stability hypothesis}
\end{subfigure}\hfill
\begin{subfigure}[t]{0.48\linewidth}
\centering
\begin{tikzpicture}
\begin{axis}[
  width=6.2cm,
  height=4.4cm,
  xlabel={$\lambda_{\mathrm{bg}}$},
  ylabel={},
  ylabel style={rotate=-90, anchor=east, xshift=-2pt, font=\small},
  xlabel style={font=\small},
  tick label style={font=\footnotesize},
  legend style={font=\scriptsize},
  grid=major,
  grid style={line width=0.25pt, draw=gray!22},
  axis lines=left,
  every axis plot/.append style={line width=0.7pt},
  legend cell align=left,
  legend style={at={(0.5,-0.46)}, anchor=north, legend columns=2, font=\tiny, draw=none, /tikz/every even column/.append style={column sep=1ex}}
]
  \addplot[black, mark=*] table[x=lbg, y=mean_est] {ex2_detect.dat};
  \addlegendentry{mean $\hat\lambda_{\mathrm{bg}}$}
  \addplot[blue, mark=square*] table[x=lbg, y=sat_freq] {ex2_detect.dat};
  \addlegendentry{saturation frequency}
  \addplot[red, dashed, no marks] coordinates {(0.60,0) (0.60,1)};
  \addlegendentry{stability boundary}
\end{axis}
\end{tikzpicture}
\caption{Example 2: saturation as a detector}
\end{subfigure}
\caption{Behavior near the stability boundary. Left, Example~1 with
         $\lambda_{\max}\E[H]=1.5>1$: adaptive truncation restores convergence to
         $\lambda_*$. Right, Example~2: $Z$-estimator saturation detects instability
         despite the unobserved background rate.}
\label{fig:stability}
\end{figure}
\subsection{Statistical uncertainty and stopping}
\label{sec:sim_ci}

\begin{figure}[t]
\centering
\begin{subfigure}[t]{0.48\linewidth}
\centering
\begin{tikzpicture}
\begin{axis}[
  width=6.2cm,
  height=4.4cm,
  xlabel={iteration $k$},
  ylabel={$\lambda^{\step{k}}$},
  ylabel style={rotate=-90, anchor=east, xshift=-2pt, font=\small},
  xlabel style={font=\small},
  tick label style={font=\footnotesize},
  legend style={font=\scriptsize},
  grid=major,
  grid style={line width=0.25pt, draw=gray!22},
  axis lines=left,
  every axis plot/.append style={line width=0.7pt},
  xmin=0, xmax=120,
  ymin=0.28, ymax=0.82
]
  \addplot[name path=lo, draw=none, forget plot] table[x=k, y=lo] {ex1_ci_M.dat};
  \addplot[name path=hi, draw=none, forget plot] table[x=k, y=hi] {ex1_ci_M.dat};
  \addplot[blue!12] fill between[of=lo and hi];
  \addplot[black, line width=0.6pt, no marks] table[x=k, y=lambda] {ex1_ci_M.dat};
  \addplot[red, very thick, dashed, no marks] coordinates {(0,0.5310) (120,0.5310)};
\end{axis}
\end{tikzpicture}
\caption{$\mathrm{M}$, $\E[H^4]<\infty$}
\end{subfigure}\hfill
\begin{subfigure}[t]{0.48\linewidth}
\centering
\begin{tikzpicture}
\begin{axis}[
  width=6.2cm,
  height=4.4cm,
  xlabel={iteration $k$},
  ylabel={$\lambda^{\step{k}}$},
  ylabel style={rotate=-90, anchor=east, xshift=-2pt, font=\small},
  xlabel style={font=\small},
  tick label style={font=\footnotesize},
  legend style={font=\scriptsize},
  grid=major,
  grid style={line width=0.25pt, draw=gray!22},
  axis lines=left,
  every axis plot/.append style={line width=0.7pt},
  xmin=0, xmax=120,
  ymin=0.28, ymax=0.82
]
  \addplot[name path=lo, draw=none, forget plot] table[x=k, y=lo] {ex1_ci_H2_c4.dat};
  \addplot[name path=hi, draw=none, forget plot] table[x=k, y=hi] {ex1_ci_H2_c4.dat};
  \addplot[blue!12] fill between[of=lo and hi];
  \addplot[black, line width=0.6pt, no marks] table[x=k, y=lambda] {ex1_ci_H2_c4.dat};
  \addplot[red, very thick, dashed, no marks] coordinates {(0,0.4038) (120,0.4038)};
\end{axis}
\end{tikzpicture}
\caption{$\mathrm{H}_2^{(4)}$, $\E[H^4]<\infty$}
\end{subfigure}\\[1ex]
\begin{subfigure}[t]{0.48\linewidth}
\centering
\begin{tikzpicture}
\begin{axis}[
  width=6.2cm,
  height=4.4cm,
  xlabel={iteration $k$},
  ylabel={$\lambda^{\step{k}}$},
  ylabel style={rotate=-90, anchor=east, xshift=-2pt, font=\small},
  xlabel style={font=\small},
  tick label style={font=\footnotesize},
  legend style={font=\scriptsize},
  grid=major,
  grid style={line width=0.25pt, draw=gray!22},
  axis lines=left,
  every axis plot/.append style={line width=0.7pt},
  xmin=0, xmax=120,
  ymin=0.28, ymax=0.82
]
  \addplot[name path=lo, draw=none, forget plot] table[x=k, y=lo] {ex1_ci_P4.5.dat};
  \addplot[name path=hi, draw=none, forget plot] table[x=k, y=hi] {ex1_ci_P4.5.dat};
  \addplot[blue!12] fill between[of=lo and hi];
  \addplot[black, line width=0.6pt, no marks] table[x=k, y=lambda] {ex1_ci_P4.5.dat};
  \addplot[red, very thick, dashed, no marks] coordinates {(0,0.6171) (120,0.6171)};
\end{axis}
\end{tikzpicture}
\caption{$\mathrm{P}_{4.5}$, $\E[H^4]<\infty$}
\end{subfigure}\hfill
\begin{subfigure}[t]{0.48\linewidth}
\centering
\begin{tikzpicture}
\begin{axis}[
  width=6.2cm,
  height=4.4cm,
  xlabel={iteration $k$},
  ylabel={$\lambda^{\step{k}}$},
  ylabel style={rotate=-90, anchor=east, xshift=-2pt, font=\small},
  xlabel style={font=\small},
  tick label style={font=\footnotesize},
  legend style={font=\scriptsize},
  grid=major,
  grid style={line width=0.25pt, draw=gray!22},
  axis lines=left,
  every axis plot/.append style={line width=0.7pt},
  xmin=0, xmax=120,
  ymin=0.28, ymax=0.82
]
  \addplot[name path=lo, draw=none, forget plot] table[x=k, y=lo] {ex1_ci_P2.5.dat};
  \addplot[name path=hi, draw=none, forget plot] table[x=k, y=hi] {ex1_ci_P2.5.dat};
  \addplot[blue!12] fill between[of=lo and hi];
  \addplot[black, line width=0.6pt, no marks] table[x=k, y=lambda] {ex1_ci_P2.5.dat};
  \addplot[red, very thick, dashed, no marks] coordinates {(0,0.5453) (120,0.5453)};
\end{axis}
\end{tikzpicture}
\caption{$\mathrm{P}_{2.5}$, $\E[H^4]=\infty$}
\end{subfigure}\\[0.5ex]
\input{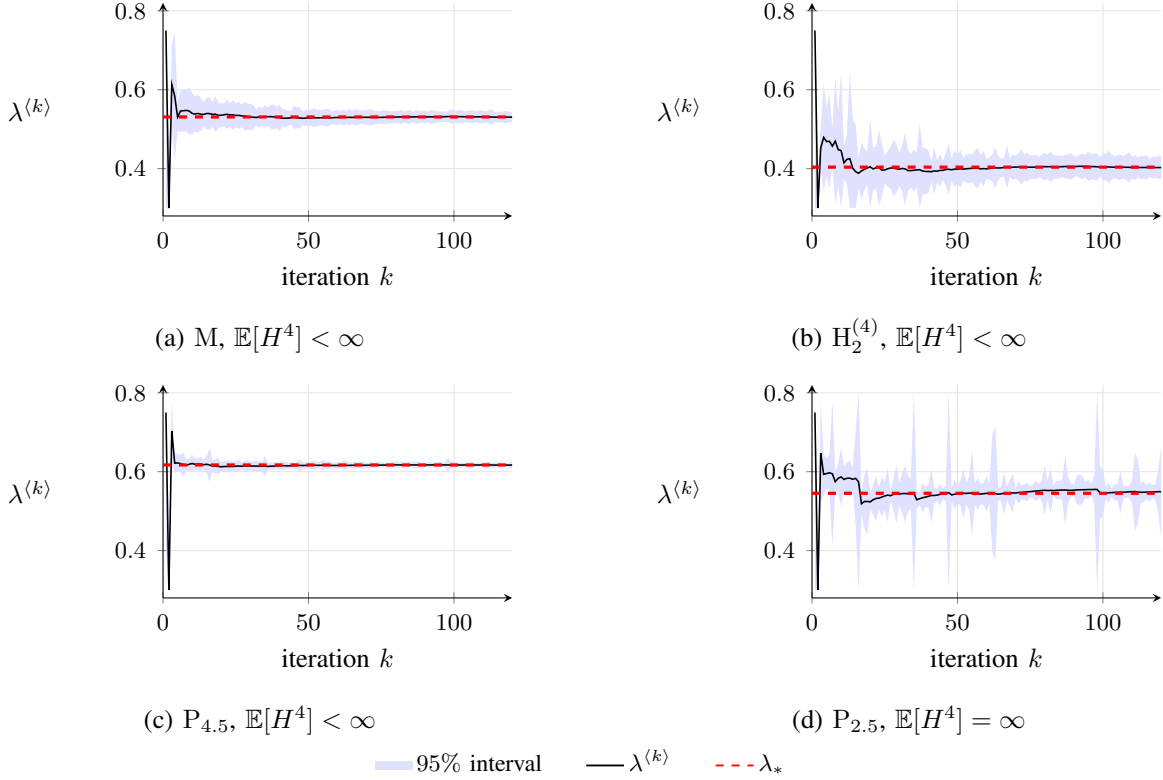}
\caption{Example 1: delta-method confidence intervals expressed as
         approximate intervals for $\lambda_*$. The bands contract at the
         parametric rate, but coverage deteriorates for $\mathrm{P}_{2.5}$, where
         the fourth-moment condition fails.}
\label{fig:ex1_ci}
\end{figure}

We next examine the statistical uncertainty of the gradient estimator and
its use for assessing approximate optimality. Since the service times are
recovered from the observed timestamps, $\hat{\bm{\theta}}^{\step{k}}$ is
based on i.i.d.\ observations, so the central limit theorem and delta method
yield the confidence interval \eqref{eq:CI-A}. The construction, which
requires $\E[H^4]<\infty$, is given in Appendix~\ref{app:delta}.

Figure~\ref{fig:ex1_ci} shows the resulting intervals along a representative
run. Dividing their half-width by the local curvature $\mu$ in
\eqref{eq:curvature-def} gives an approximate confidence region for
$\lambda_*$. The bands contract at the parametric rate as the sample size
increases, and the iterates settle within them. The fourth-moment condition
affects coverage rather than the shrinkage rate: coverage is close to the
nominal $95\%$ level for the families with $\E[H^4]<\infty$, but deteriorates
substantially for $\mathrm{P}_{2.5}$, where $\E[H^4]=\infty$. Thus, for heavy-tailed service times, the nominal $n^{-1/2}$ confidence-band width can substantially understate the actual uncertainty when the moment conditions required for the delta-method approximation fail.

\begin{table}[t]
\centering
\caption{Service times $\mathrm{P}_{2.5}$, for which
         $\E[H^3]=\infty$. Scaled variances $n\Var$ of the estimator components.
         The inputs used in Example~3 have stable scaled variances, whereas the
         sample second moment used in Example~1 has infinite variance.}
\label{tab:ex3_moments}
\footnotesize
\begin{tabular}{rrrr}
\hline
$n$ & $n\Var[\hat f_H^*]$ & $n\Var[\hat f_H']$ & $n\Var[\widehat{\E[H^2]}]$ \\
\hline
500 & 0.2480 & 0.2194 & 685.5394 \\
2000 & 0.2318 & 0.1947 & 510.9883 \\
8000 & 0.2518 & 0.2215 & 808.5836 \\
\hline
\end{tabular}

\end{table}

\begin{table}[t]
\centering
\caption{Example 1: cost of certified stopping. The algorithm stops when
         the delta-method interval for $g+\kappa'$ lies inside
         $[-\eta\mu,\eta\mu]$, approximately certifying
         $|\lambda^{\step{k}}-\lambda_*|\lesssim\eta$.}
\label{tab:ex1_stopping}
\footnotesize
\begin{tabular}{lrrrr}
\hline
$\eta$ & $k$ at stop & packets & AoI gap & fraction stopped \\
\hline
0.02 & -- & -- & -- & 0 \\
0.05 & 54.5000 & 75625 & $1.35\times10^{-3}$ & 1 \\
0.10 & 6 & 1050 & $1.84\times10^{-2}$ & 1 \\
full horizon & 250 & 1568750 & $1.15\times10^{-4}$ & 1 \\
\hline
\end{tabular}

\end{table}

The contrast with Example~3 is particularly clear for heavy-tailed service
times. There, the estimator inputs are bounded by exponential
inter-generation gaps and therefore have finite moments, whereas the
Example~1 estimator involves the sample second moment of $H$. Under
$\mathrm{P}_{2.5}$, Table~\ref{tab:ex3_moments} shows the resulting
difference in scaled variances, illustrating the robustness gained through
this censoring mechanism.

The same intervals can be used for statistical stopping. We stop when the
interval for $g+\kappa'$ lies inside
$[-\eta\mu,\eta\mu]$, which under the local curvature approximation
certifies $|\lambda^{\step{k}}-\lambda_*|\lesssim\eta$. Table~\ref{tab:ex1_stopping}
shows the resulting sample cost. Since the interval width is
$O(n^{-1/2})$, achieving accuracy $\eta$ requires $O(\eta^{-2})$ samples at
a single iteration; with $n^{\step{k}}$ increasing linearly in $k$, the
cumulative cost grows substantially faster. Thus, very precise statistical
certification can be much more expensive than obtaining a moderately
accurate solution.

\subsection{Comparison with a model-free reference method}
\label{sec:sim_bench}

We finally compare the proposed method with a generic model-free
optimization procedure. The goal is to assess whether exploiting the
analytical gradient representation in \eqref{eq:g(lambda,theta)-def} yields
a practical advantage over estimating the gradient through perturbations of
the operating rate.

As a model-free reference, we use the {\it Kiefer--Wolfowitz algorithm} \cite{KW}. At
iteration $k$, the system is run at $\lambda^{\step{k}}\pm c^{\step{k}}$
using $n^{\step{k}}/2$ informative packets at each rate, and the iterate is
updated according to, using the standard Kiefer--Wolfowitz scaling
$c^{\step{k}}=O(k^{-1/6})$,
\[
    \lambda^{\step{k+1}}
    =
    \Pi_{\mathcal L}\!\left[
    \lambda^{\step{k}}
    -\gamma^{\step{k}}
    \frac{
        \widehat{Q}_{+}^{\step{k}}
        -
        \widehat{Q}_{-}^{\step{k}}
    }{
        2c^{\step{k}}
    }
    \right],
    \qquad
    c^{\step{k}}=c_3 k^{-1/6}.
\]
The constant $c_3$ is tuned separately for each
comparison over a fixed grid. Both methods use the same step-size and
sample-size sequences and are compared at the same cumulative packet
budget. Thus, the model-free method uses the same number of informative
packets, but splits them between two perturbed operating rates.

\begin{figure}[t]
\centering
{\footnotesize$\displaystyle
 \mathrm{MSE}^{\step{k}}
 $}\\[0.6ex]
\begin{subfigure}[t]{0.32\linewidth}
\centering
\begin{tikzpicture}
\begin{axis}[
  width=6.2cm,
  height=4.4cm,
  xlabel={packets consumed $N$},
  ylabel={},
  ylabel style={rotate=-90, anchor=east, xshift=-2pt, font=\small},
  xlabel style={font=\small},
  tick label style={font=\footnotesize},
  legend style={font=\scriptsize},
  grid=major,
  grid style={line width=0.25pt, draw=gray!22},
  axis lines=left,
  every axis plot/.append style={line width=0.7pt},
  xmode=log,
  ymode=log,
  width=0.97\linewidth,
  height=0.86\linewidth,
  tick label style={font=\scriptsize},
  xlabel style={font=\footnotesize}
]
  \addplot[black, no marks] table[x=N, y=mse_sgd] {ex1_bench_M.dat};
  \addplot[blue, dashed, no marks] table[x=N, y=mse_kw] {ex1_bench_M.dat};
\end{axis}
\end{tikzpicture}
\caption{$\mathrm{M}$, $c^2=1$}
\end{subfigure}\hfill
\begin{subfigure}[t]{0.32\linewidth}
\centering
\begin{tikzpicture}
\begin{axis}[
  width=6.2cm,
  height=4.4cm,
  xlabel={packets consumed $N$},
  ylabel={},
  ylabel style={rotate=-90, anchor=east, xshift=-2pt, font=\small},
  xlabel style={font=\small},
  tick label style={font=\footnotesize},
  legend style={font=\scriptsize},
  grid=major,
  grid style={line width=0.25pt, draw=gray!22},
  axis lines=left,
  every axis plot/.append style={line width=0.7pt},
  xmode=log,
  ymode=log,
  width=0.97\linewidth,
  height=0.86\linewidth,
  tick label style={font=\scriptsize},
  xlabel style={font=\footnotesize}
]
  \addplot[black, no marks] table[x=N, y=mse_sgd] {ex1_bench_H2_c4.dat};
  \addplot[blue, dashed, no marks] table[x=N, y=mse_kw] {ex1_bench_H2_c4.dat};
\end{axis}
\end{tikzpicture}
\caption{$\mathrm{H}_2^{(4)}$, $c^2=4$}
\end{subfigure}\hfill
\begin{subfigure}[t]{0.32\linewidth}
\centering
\begin{tikzpicture}
\begin{axis}[
  width=6.2cm,
  height=4.4cm,
  xlabel={packets consumed $N$},
  ylabel={},
  ylabel style={rotate=-90, anchor=east, xshift=-2pt, font=\small},
  xlabel style={font=\small},
  tick label style={font=\footnotesize},
  legend style={font=\scriptsize},
  grid=major,
  grid style={line width=0.25pt, draw=gray!22},
  axis lines=left,
  every axis plot/.append style={line width=0.7pt},
  xmode=log,
  ymode=log,
  width=0.97\linewidth,
  height=0.86\linewidth,
  tick label style={font=\scriptsize},
  xlabel style={font=\footnotesize}
]
  \addplot[black, no marks] table[x=N, y=mse_sgd] {ex1_bench_LN.dat};
  \addplot[blue, dashed, no marks] table[x=N, y=mse_kw] {ex1_bench_LN.dat};
\end{axis}
\end{tikzpicture}
\caption{$\mathrm{LN}$, $c^2=4$}
\end{subfigure}\\[0.5ex]
\input{bench_legend.tex}
\caption{Example 1: mean squared error of the iterate against the packet
         budget, over $R=60$ replications, for three service-time families with
         $\E[H]=1$. Both methods use identical step-size and sample-size sequences,
         and the model-free method is shown at its best perturbation width on a
         grid.}
\label{fig:ex1_bench}
{\footnotesize$\displaystyle
 \mathrm{MSE}^{\step{k}}
 $}\\[0.6ex]
\begin{subfigure}[t]{0.32\linewidth}
\centering
\begin{tikzpicture}
\begin{axis}[
  width=6.2cm,
  height=4.4cm,
  xlabel={packets consumed $N$},
  ylabel={},
  ylabel style={rotate=-90, anchor=east, xshift=-2pt, font=\small},
  xlabel style={font=\small},
  tick label style={font=\footnotesize},
  legend style={font=\scriptsize},
  grid=major,
  grid style={line width=0.25pt, draw=gray!22},
  axis lines=left,
  every axis plot/.append style={line width=0.7pt},
  xmode=log,
  ymode=log,
  width=0.97\linewidth,
  height=0.86\linewidth,
  tick label style={font=\scriptsize},
  xlabel style={font=\footnotesize}
]
  \addplot[black, no marks] table[x=N, y=mse_sgd] {ex2_bench.dat};
  \addplot[blue, dashed, no marks] table[x=N, y=mse_kw] {ex2_bench.dat};
\end{axis}
\end{tikzpicture}
\caption{Example 2, $\lambda_{\mathrm{bg}}=0.3$}
\end{subfigure}\hfill
\begin{subfigure}[t]{0.32\linewidth}
\centering
\begin{tikzpicture}
\begin{axis}[
  width=6.2cm,
  height=4.4cm,
  xlabel={packets consumed $N$},
  ylabel={},
  ylabel style={rotate=-90, anchor=east, xshift=-2pt, font=\small},
  xlabel style={font=\small},
  tick label style={font=\footnotesize},
  legend style={font=\scriptsize},
  grid=major,
  grid style={line width=0.25pt, draw=gray!22},
  axis lines=left,
  every axis plot/.append style={line width=0.7pt},
  xmode=log,
  ymode=log,
  width=0.97\linewidth,
  height=0.86\linewidth,
  tick label style={font=\scriptsize},
  xlabel style={font=\footnotesize}
]
  \addplot[black, no marks] table[x=N, y=mse_sgd] {ex3_bench.dat};
  \addplot[blue, dashed, no marks] table[x=N, y=mse_kw] {ex3_bench.dat};
\end{axis}
\end{tikzpicture}
\caption{Example 3, $\kappa(\lambda)=0.25\lambda$}
\end{subfigure}\hfill
\begin{subfigure}[t]{0.32\linewidth}
\centering
\begin{tikzpicture}
\begin{axis}[
  width=6.2cm,
  height=4.4cm,
  xlabel={packets consumed $N$},
  ylabel={},
  ylabel style={rotate=-90, anchor=east, xshift=-2pt, font=\small},
  xlabel style={font=\small},
  tick label style={font=\footnotesize},
  legend style={font=\scriptsize},
  grid=major,
  grid style={line width=0.25pt, draw=gray!22},
  axis lines=left,
  every axis plot/.append style={line width=0.7pt},
  xmode=log,
  ymode=log,
  width=0.97\linewidth,
  height=0.86\linewidth,
  tick label style={font=\scriptsize},
  xlabel style={font=\footnotesize}
]
  \addplot[black, no marks] table[x=N, y=mse_sgd] {ex4_bench.dat};
  \addplot[blue, dashed, no marks] table[x=N, y=mse_kw] {ex4_bench.dat};
\end{axis}
\end{tikzpicture}
\caption{Example 4, Erlang-4 service}
\end{subfigure}\\[0.5ex]
\input{bench_legend.tex}
\caption{Examples 2--4: mean squared error of the iterate against the packet
         budget, over $R=50$ replications, at equal budget and with the model-free
         method at its best perturbation width on a grid.}
\label{fig:bench_234}
\end{figure}

Figures~\ref{fig:ex1_bench} and~\ref{fig:bench_234} show the resulting
iterate errors across all four examples in terms of the mean squared error, $\mathrm{MSE}^{\step{k}}=\frac{1}{R}\sum_{r=1}^{R}\bigl(\lambda_r^{\step{k}}-\lambda_*\bigr)^2$. The proposed method achieves
substantially higher terminal accuracy in every example. The magnitude of
the improvement depends on the problem structure: in Example~2, the nearby
stability boundary makes finite-difference estimation particularly noisy;
in Example~4, the objective is very flat near the optimum; and in Example~3,
where no stability constraint limits the perturbation, the model-free method
performs better but remains less accurate. In Example~1, the advantage of
the proposed method becomes more pronounced as service-time variability
increases. Deterministic service is excluded from the benchmark because the
structural gradient estimator is then exact, making the comparison
essentially one between exact projected gradient descent and a noisy
finite-difference method.

\section{Discussion and concluding remarks}
\label{sec:conclusion}

This paper developed a structure-aware approach to finding the optimal update
generation rate for online AoI optimization when the
underlying queueing and service statistics are unknown. The key idea is to
exploit an analytical representation of the AoI gradient and estimate only
the finite-dimensional quantities appearing in that representation. This
enables direct online optimization of the update generation rate from
operational observations, without requiring prior knowledge of the system
parameters or learning the entire objective as a black box.

The numerical results demonstrate the practical value of this approach.
Across four qualitatively different incomplete-information settings, the
proposed algorithm converges to the optimal operating point and remains
effective under service-time variability, proximity to stability boundaries,
hidden traffic, and partial observations. Moreover, at the same observation
budget, it achieves substantially smaller optimization error than the
Kiefer--Wolfowitz benchmark in the considered numerical settings, illustrating
how exploiting queueing structure can improve the efficiency of online
optimization.

Several directions remain open, including extensions to time-varying system
parameters, interacting queues, and approximate gradient representations.
The statistical stopping results also highlight the gap between finding a
good operating point and certifying near-optimality, motivating efficient
certification procedures under weaker moment assumptions.

More broadly, the results suggest that data-driven optimization need not be
strictly model-free. When analytical structure provides a tractable
representation of the optimization problem, combining that structure with
online statistical estimation can retain the adaptability of data-driven
methods while exploiting information that black-box approaches discard.
This points toward a broader class of model-informed, data-driven
optimization methods that use system knowledge not to eliminate uncertainty,
but to extract useful structure from it.
\appendices

\section{Proof of Theorem \ref{thm:FD_convex}}
\label{appendix:FD_convex}

The proof of Theorem \ref{thm:FD_convex} uses the following inequalities.
\begin{lemma} \label{lem:ineq}
For any non-negative random variable $Y$ with $\E[Y] \leq 1$, we have
\begin{align}
\E[Ye^{-Y}] &\leq \E[e^{-Y}],
\label{eq:YexpY-bound}
\\
\E[Y^2e^{-Y}] &\leq \E[e^{-Y}].
\label{eq:Y2expY-bound}
\end{align}
\end{lemma}

\begin{proof}
Since $e^{-Y}$ is a decreasing function of $Y$, the random variables
$Y$ and $e^{-Y}$ have non-positive covariance (see, e.g., \cite{Behboodian1994}). Hence,
\[
    \E[Ye^{-Y}]
    \leq
    \E[Y]\E[e^{-Y}]
    \leq
    \E[e^{-Y}],
\]
where the last inequality follows from $\E[Y]\leq1$. This proves
\eqref{eq:YexpY-bound}.

To prove \eqref{eq:Y2expY-bound}, it suffices to show that
$
\E[(Y^2-1)e^{-Y}] \leq 0.$
Observe that
\begin{equation}
(y^2-1)e^{-y}\leq 2e^{-1}(y-1),
\qquad y\ge0.
\label{eq:linear-UB}
\end{equation}
Indeed,
\[
    \frac{(y^2-1)e^{-y}}{y-1}=(y+1)e^{-y},
\]
and the function $\gamma(y):=(y+1)e^{-y}$ is decreasing on $[0,\infty)$ since
$
\gamma'(y)=-ye^{-y}\le0.
$
Therefore,
\[
    (y+1)e^{-y}
    \begin{cases}
    \ge 2e^{-1}, & 0\le y\le1,\\
    \le 2e^{-1}, & y>1,
    \end{cases}
\]
which is equivalent to \eqref{eq:linear-UB}. Consequently,
\[
    \E[(Y^2-1)e^{-Y}]
    \le
    2e^{-1}\E[Y-1]
    =
    2e^{-1}(\E[Y]-1)
    \le0,
\]
thus establishing \eqref{eq:Y2expY-bound}.
\end{proof}

We then provide the proof of Theorem \ref{thm:FD_convex}.
We rewrite the right-hand side of \eqref{eq:EA} as
$
\E_{\lambda}[A] = \psi(\lambda) + \phi(\lambda),
$
where
\begin{equation}
\psi(\lambda):=\frac{\lambda\E[H^{2}]}{2\left(1-\lambda\E[H]\right)}+\E[H],
\quad
\phi(\lambda):=\frac{1-\lambda\E[H]}{\lambda f^{*}_{H}(\lambda)}.
\label{eq:psi-phi-def}
\end{equation}
It is easy to see that $\psi(\lambda)$ is strictly convex on $\left(0,1/\E[H]\right)$: the first derivative of $\psi(\lambda)$ is given by
\begin{align*}
\psi^\prime(\lambda) 
&=
\frac{\E[H^2]}{2(1-\lambda\,\E[H])^2},
\end{align*}
which is strictly increasing for $\lambda\in\left(0,{1}/{\E[H]}\right)$.
Therefore, we prove the convexity of $\phi(\lambda)$ below.

The first derivative is
\begin{align*}
\phi'(\lambda)
&=
\frac{
    -f_H^*(\lambda)-\lambda(1-\lambda\E[H])f_H'(\lambda)
}{
    \lambda^2\{f_H^*(\lambda)\}^2
}
=
-\phi(\lambda)
\left(
\frac{1}{\lambda(1-\lambda\E[H])}
+
\frac{f_H'(\lambda)}{f_H^*(\lambda)}
\right).
\end{align*}
Differentiating once more and simplifying yields
\begin{align*}
\phi''(\lambda)
&=
\phi(\lambda)
\left[
\left(
\frac{1}{\lambda(1-\lambda\E[H])}
+
\frac{f_H'(\lambda)}{f_H^*(\lambda)}
\right)^2
+
\frac{1-2\lambda\E[H]}
{\lambda^2(1-\lambda\E[H])^2}
-
\frac{f_H''(\lambda)}{f_H^*(\lambda)}
+
\left(
\frac{f_H'(\lambda)}{f_H^*(\lambda)}
\right)^2
\right]
\\
&=
\phi(\lambda)
\left[
\frac{2\E[H]}{1-\lambda\E[H]}
\left(
\frac{1}{\lambda}
+
\frac{f_H'(\lambda)}{f_H^*(\lambda)}
\right)
+
\left(
\frac{1}{\lambda}
+
\frac{f_H'(\lambda)}{f_H^*(\lambda)}
\right)^2
+
\frac{1}{\lambda^2}
-
\frac{f_H''(\lambda)}{f_H^*(\lambda)}
+
\left(
\frac{f_H'(\lambda)}{f_H^*(\lambda)}
\right)^2
\right].
\end{align*}
Therefore, it suffices to prove
\[
    \frac{f_H''(\lambda)}{f_H^*(\lambda)}
    \le
    \frac{1}{\lambda^2},
    \qquad
    -\frac{f_H'(\lambda)}{f_H^*(\lambda)}
    \le
    \frac{1}{\lambda}.
\]
Since
\[
    f_H^*(\lambda)=\E[e^{-\lambda H}],\qquad
    f_H'(\lambda)=-\E[He^{-\lambda H}],\qquad
    f_H''(\lambda)=\E[H^2e^{-\lambda H}],
\]
these inequalities are equivalent to
\[
    \E[(\lambda H)^2e^{-\lambda H}]
    \le
    \E[e^{-\lambda H}],
    \qquad
    \E[\lambda He^{-\lambda H}]
    \le
    \E[e^{-\lambda H}].
\]
Setting $Y=\lambda H$, both inequalities follow immediately from Lemma~\ref{lem:ineq}.
\qed

\section{Proof of Theorem \ref{thm:FB_convex}}
\label{appendix:FB_convex}

We use the following corollary to prove Theorem \ref{thm:FB_convex}.
\begin{corollary}
Under the stability condition (\ref{eq:stability}), we have
\begin{align}
\zeta(\lambda,\lambda_\bg)+\lambda \zeta'(\lambda,\lambda_\bg) &\geq 0,
\label{eq:zeta-ineq-1}\\
\zeta(\lambda,\lambda_\bg) - \lambda^2 \zeta''(\lambda) &\geq 0.
\label{eq:zeta-ineq-2}
\end{align}
\end{corollary}
\begin{proof}
Under the stability condition \eqref{eq:stability},
\[
    \lambda \E[B_{\bg}]
    =
    \lambda\frac{\E[H]}{1-\lambda_{\bg}\E[H]}
    <1.
\]
Using \eqref{eq:zeta-by-f_B_bg}, we have 
$
\zeta(\lambda,\lambda_{\bg})
=
f_{B_{\bg}}^*(\lambda).$
Together with
\[
    f_{B_{\bg}}'(\lambda)
    =
    -\E[B_{\bg}e^{-\lambda B_{\bg}}],
    \qquad
    f_{B_{\bg}}''(\lambda)
    =
    \E[B_{\bg}^2e^{-\lambda B_{\bg}}],
\]
we obtain
\begin{align*}
\zeta(\lambda,\lambda_{\bg})
+\lambda\zeta'(\lambda,\lambda_{\bg})
&=
\E[e^{-\lambda B_{\bg}}]
-
\E[\lambda B_{\bg}e^{-\lambda B_{\bg}}]
\ge 0,
\\
\zeta(\lambda,\lambda_{\bg})
-\lambda^2\zeta''(\lambda,\lambda_{\bg})
&=
\E[e^{-\lambda B_{\bg}}]
-
\E[(\lambda B_{\bg})^2e^{-\lambda B_{\bg}}]
\ge 0.
\end{align*}
where both inequalities follow from Lemma~\ref{lem:ineq}.
\end{proof}

To prove the convexity of $\E_{\lambda}[A]$, 
we rewrite the right-hand side of \eqref{eq:EA_bg} as
\[
    \E_{\lambda}[A] = \psi(\lambda,\lambda_{\bg}) + 
    \frac{\lambda_{\bg}\E[H]}{\lambda}
    +
    \phi(\lambda,\lambda_{\bg}),
\]
where
\[
    \psi(\lambda,\lambda_{\bg}):=\frac{(\lambda+\lambda_{\bg})\E[H^{2}]}{2\left\{1-(\lambda+\lambda_{\bg})\E[H]\right\}}+\E[H],
    \quad
    \phi(\lambda, \lambda_{\bg}):=\frac{1-(\lambda+\lambda_{\bg})\E[H]}{\lambda \zeta(\lambda,\lambda_{\bg})}.
\]

Note that we have $\psi(\lambda,\lambda_{\bg})=\psi(\lambda+\lambda_{\bg})$ (cf.\ \eqref{eq:psi-phi-def}),
so that it is a strictly convex function of $\lambda$ as shown in Appendix \ref{appendix:FD_convex}.
Also, $\lambda_{\bg}\E[H]/\lambda$ is obviously convex with respect to $\lambda$.
Therefore, to prove the convexity of $\E_{\lambda}[A]$, it is sufficient to
show the convexity of
$\phi(\lambda,\lambda_{\bg})$

Recall that $\zeta'(\lambda,\lambda_{\bg})$ and 
$\zeta''(\lambda,\lambda_{\bg})$ denote the first and second partial derivatives of $\zeta(\lambda,\lambda_{\bg})$ with respect to $\lambda$. Similarly, let 
$\phi'(\lambda,\lambda_{\bg})$ and $\phi''(\lambda,\lambda_{\bg})$ denote 
the first and second partial derivatives of $\phi(\lambda,\lambda_{\bg})$ with respect to $\lambda$.
With similar calculations as in Appendix \ref{appendix:FD_convex}, we can show that
\begin{align*}
\phi'(\lambda,\lambda_{\bg}) 
&= 
-\phi(\lambda,\lambda_{\bg})
\left[
\frac{1-\lambda_{\bg}\E[H]}
{\lambda\{1-(\lambda+\lambda_{\bg})\E[H]\}}
+
\frac{\zeta'(\lambda,\lambda_{\bg})}
{\zeta(\lambda,\lambda_{\bg})}
\right],
\\
\phi''(\lambda,\lambda_{\bg})
&=
\phi(\lambda,\lambda_{\bg})
\Biggl[
\frac{2\E[H]}
{1-(\lambda+\lambda_{\bg})\E[H]}
\left(
\frac{1}{\lambda}
+
\frac{\zeta'(\lambda,\lambda_{\bg})}
{\zeta(\lambda,\lambda_{\bg})}
\right)
\\
&
\quad\qquad\qquad
{}+
\left(
\frac{1}{\lambda}
+
\frac{\zeta'(\lambda,\lambda_{\bg})}
{\zeta(\lambda,\lambda_{\bg})}
\right)^2
+
\frac{1}{\lambda^2}
-
\frac{\zeta''(\lambda,\lambda_{\bg})}
{\zeta(\lambda,\lambda_{\bg})}
+
\left(
\frac{\zeta'(\lambda,\lambda_{\bg})}
{\zeta(\lambda,\lambda_{\bg})}
\right)^2
\Biggr].
\end{align*}
We can then verify 
$\phi''(\lambda) \geq 0$ from (\ref{eq:zeta-ineq-1}) and (\ref{eq:zeta-ineq-2}).
\qed

\section{Proof of Lemmas in Section~\ref{sec:application}}
\label{appendix:regularity}

\begin{proof}[Proof of Lemma~\ref{lem:g_lip}]
{Recalling 
 $
 {\boldsymbol \theta}_\lambda\equiv (\E[H],\E[H^2],f_H^*(\lambda),f_H'(\lambda))$,
 \[
     g(\lambda;\boldsymbol{\theta})
     =
     \frac{\theta_2}{2(1-\lambda\theta_1)^2}
     -\frac{\theta_4}{\lambda\theta_3^2}
     -\frac{1}{\lambda^2\theta_3}
     +\frac{\theta_1\theta_4}{\theta_3^2}.
 \]
 Since $\Theta_\epsilon$ is compact and $\mathcal L=[\lambda_{\min},\lambda_{\max}]$
 with $\lambda_{\min}>0$, there exists a constant $M_\Theta<\infty$ such that
 $
 |\theta_i|\le M_\Theta$ for all $i=1,\ldots,4$ and 
 for all $\boldsymbol{\theta}\in\Theta_\epsilon$. Moreover, by construction of
 $\Theta_\epsilon$,
 \[
     1-\lambda\theta_1\ge\lambda\epsilon\ge\lambda_{\min}\epsilon,
     \qquad
     \theta_3\ge\epsilon,
 \]
 for every $\lambda\in\mathcal L$. Hence all denominators appearing in
 $g(\lambda;\boldsymbol{\theta})$ are uniformly bounded away from zero.

 The partial derivatives of $g$ with respect to ${\boldsymbol\theta}$ are
 \[
     \begin{aligned}
     \frac{\partial g}{\partial\theta_1}
     &=
     \frac{\lambda\theta_2}{(1-\lambda\theta_1)^3}
     +\frac{\theta_4}{\theta_3^2},
     &\qquad
     \frac{\partial g}{\partial\theta_2}
     &=
     \frac{1}{2(1-\lambda\theta_1)^2},
     \\[1ex]
     \frac{\partial g}{\partial\theta_3}
     &=
     \frac{2\theta_4}{\lambda\theta_3^3}
     +\frac{1}{\lambda^2\theta_3^2}
     -\frac{2\theta_1\theta_4}{\theta_3^3},
     &\qquad
     \frac{\partial g}{\partial\theta_4}
     &=
     -\frac{1}{\lambda\theta_3^2}
     +\frac{\theta_1}{\theta_3^2}.
     \end{aligned}
 \]
 Since $\lambda\in\mathcal L$, the variables $\theta_i$ are uniformly bounded on
 $\Theta_\epsilon$, and all denominators are uniformly bounded away from zero,
 each partial derivative is uniformly bounded. Therefore there exists a constant
 $C_1<\infty$ such that
 \[
     \sup_{\lambda\in\mathcal L}
     \sup_{\boldsymbol{\theta}\in\Theta_\epsilon}
     \|\nabla_{\boldsymbol{\theta}}g(\lambda;\boldsymbol{\theta})\|
     \le C_1.
 \]

 Finally, by the multivariate mean value theorem,
 \[
     |g(\lambda;\boldsymbol{\theta})
     -g(\lambda;\boldsymbol{\theta}')|
     \le
     C_1\|\boldsymbol{\theta}-\boldsymbol{\theta}'\|,
 \]
 for all $\boldsymbol{\theta},\boldsymbol{\theta}'\in\Theta_\epsilon$ and all
 $\lambda\in\mathcal L$, proving \eqref{eq:Lips}.
}
\end{proof}

\begin{proof}[Proof of Lemma~\ref{lem:theta_FD}]
By the triangle inequality, bearing in mind $
{\boldsymbol \theta}\equiv (\E[H],\E[H^2],f_H^*(\lambda),f_H'(\lambda))$,
\begin{align}
\mathbb{E}_{\lambda^{\step{k}}}\!\left[ \|\bm{\theta}_{\lambda^{\step{k}}} - \hat{\bm{\theta}}^{\step{k}}\|^2 \right]
\leq\,&
\E\!\left[\left|\hat{h}_{1,n^{\step{k}}}-\E[H]\right|^2\right]
+\E\!\left[\left|\hat{h}_{2,n^{\step{k}}}-\E[H^2]\right|^2\right] \notag \\
&
{}+
\E\!\left[\left|\hat{h}_{3,n^{\step{k}}}(\lambda)-\E[e^{-\lambda H}]\right|^2\right]
+\E\!\left[\left|\hat{h}_{4,n^{\step{k}}}(\lambda)-\E[-He^{-\lambda H}]\right|^2\right]. \label{eq:trian}
\end{align}
Therefore, it suffices to verify the required condition for each term on the right-hand side separately. Since the observations $(H_i)_{i\geq 1}$ are i.i.d., we have for any $n\in{\mathbb N}$,
\begin{equation}\label{eq:sigma_H}
\E\left[\left|\frac{1}{n}\sum_{i=1}^n \varphi(H_i)-\E[\varphi(H)]\right|^2\right]
=\frac{\Var\!\left(\varphi(H)\right)}{n},
\end{equation}
for any square-integrable function $\varphi(\cdot)$. Moreover, if $\E[H^4]<\infty$, then 
\begin{align}
\tilde{C}_2
:=
\sup_{\lambda\in[\lambda_{\min},\lambda_{\max}]}
\max\left\{
\Var(H),\,
\Var(H^2),\,
\Var(e^{-\lambda H}),\,
\Var(-He^{-\lambda H})
\right\}
<\infty.
\end{align}
Define the event
\begin{align*}
A_\epsilon
:={}&
\left\{
\frac{1}{n^{\step{k}}}\sum_{i=1}^{n^{\step{k}}} H_i
>
\frac{1}{\lambda_{\max}}-\epsilon
\right\}
\cup
\left\{
\frac{1}{n^{\step{k}}}\sum_{i=1}^{n^{\step{k}}} H_i^2
>
K-\epsilon
\right\} \cup
\left\{
\frac{1}{n^{\step{k}}}\sum_{i=1}^{n^{\step{k}}}
e^{-\lambda_{\max}H_i}
<
\epsilon
\right\},
\end{align*}
and observe that the estimator is truncated only on the event \(A_\epsilon\). By Boole's inequality,
\begin{align*}
\E[\mathbf{1}\{A_\epsilon\}]
&\leq 
\mathbb{P}\left(
\frac{1}{n^{\step{k}}}\sum_{i=1}^{n^{\step{k}}}H_i>\frac{1}{\lambda_{\max}}-\epsilon
\right)
+\mathbb{P}\left(
\frac{1}{n^{\step{k}}}\sum_{i=1}^{n^{\step{k}}}H_i^2>K-\epsilon
\right)\\
&\quad+
\mathbb{P}\left(
\frac{1}{n^{\step{k}}}\sum_{i=1}^{n^{\step{k}}}e^{-\lambda_{\max}H_i}<\epsilon
\right).
\end{align*}
We proceed by identifying upper bound on the three probabilities in the right-hand side of the previous display, thus yielding an upper bound on $\E[\mathbf{1}\{A_\epsilon\}]$. 
\begin{itemize}
\item By \eqref{eq:epsbound} we have $2\epsilon<\lambda_{\max}^{-1}-\E[H]$, which is equivalent to $\E[H]+2\epsilon<\lambda_{\max}^{-1}$. Therefore,
\begin{align*}
\mathbb{P}\left(\frac{1}{n^{\step{k}}}\sum_{i=1}^{n^{\step{k}}}H_i>\frac{1}{\lambda_{\max}}-\epsilon\right)\leq 
\mathbb{P}\left(\left| \frac{1}{n^{\step{k}}}\sum_{i=1}^{n^{\step{k}}}H_i-\E[H]\right|> \epsilon\right).
\end{align*}
Applying Chebyshev's inequality and \eqref{eq:sigma_H} yields
\begin{align*}
\mathbb{P}\left(\frac{1}{n^{\step{k}}}\sum_{i=1}^{n^{\step{k}}}H_i>\frac{1}{\lambda_{\max}}-\epsilon\right)\leq \frac{\Var(H)}{\epsilon^2 n^{\step{k}}}\leq \frac{\tilde{C}_2}{\epsilon^2 n^{\step{k}}}.
\end{align*}
\item Similarly, \eqref{eq:epsbound} yields $\E[H^2]<K-2\epsilon$, so that
\begin{align*}
\mathbb{P}\left(\frac{1}{n^{\step{k}}}\sum_{i=1}^{n^{\step{k}}}H_i^2>K-\epsilon\right)\leq \frac{\Var[H^2]}{\epsilon^2 n^{\step{k}}}\leq \frac{\tilde{C}_2}{\epsilon^2 n^{\step{k}}},
\end{align*}
\item Finally, $\epsilon<{f_H^*(\lambda_{\max})}/{2}$ implies $\epsilon<f_H^*(\lambda_{\max})-\epsilon$, and
\begin{align*}
\mathbb{P}\left(\frac{1}{n^{\step{k}}}\sum_{i=1}^{n^{\step{k}}}e^{-\lambda_{\max} H_i}<\epsilon\right)&\leq \mathbb{P}\left(\left|\frac{1}{n^{\step{k}}}\sum_{i=1}^{n^{\step{k}}}e^{-\lambda_{\max} H_i}-f_H^*(\lambda_{\max})\right|>\epsilon\right)\\&\leq \frac{\Var[e^{-\lambda_{\max} H}]}{\epsilon^2 n^{\step{k}}}\leq \frac{\tilde{C}_2}{\epsilon^2 n^{\step{k}}}.
\end{align*}
\end{itemize}

Observe that for $j=1,2$,
with ${\mathscr C}_j$ being an upper bound on both $\hat{h}_{j,n^{\step{k}}}^2$ and $\theta_j^2$, the standard inequality $(x-y)^2\le 2(x^2+y^2)$ yields
\begin{align}
\E\left[|\hat{h}_{j,n^{\step{k}}}-\theta_j|^2\right]&=\E[|\hat{h}_{j,n^{\step{k}}}-\theta_j|^2\mathbf{1}\{A_\epsilon\}]+\E[|\hat{h}_{j,n^{\step{k}}}-\theta_j|^2\mathbf{1}\{{A}^{\rm c}_\epsilon\}] \notag\\
&\leq 4\,{\mathscr C}_j\,\E[\mathbf{1}\{A_\epsilon\}]+\E\left[\left|\frac{1}{n^{\step{k}}}\sum_{i=1}^{n^{\step{k}}}H_i^j-\E[H^j]\right|^2\right].\label{eq:BND}
\end{align}

We now identify an upper bound on the first term in the right hand side of \eqref{eq:trian}. 
By \eqref{eq:sigma_H},
\[\E\left[\left|\frac{1}{n^{\step{k}}}\sum_{i=1}^{n^{\step{k}}}H_i-\E[H]\right|^2\right]\le \frac{\Var(H)}{n^{\step{k}}}\le\frac{\tilde {C}_2}{n^{\step{k}}}.\]
Because $\smash{\hat{h}_{1,n^{\step{k}}}\le \lambda_{\max}^{-1} -\epsilon \le \lambda_{\max}^{-1}}$ and $\theta_1=\smash{\E[H]\le \lambda_{\max}^{-1} -2\epsilon \le \lambda_{\max}^{-1}}$, we obtain from \eqref{eq:BND} that
\begin{align*}
\E\left[|\hat{h}_{1,n^{\step{k}}}-\E[H]|^2\right]
&\leq {\frac{4}{\lambda^2_{\max}}}\E[\mathbf{1}\{A_\epsilon\}]+\E\left[\left|\frac{1}{n^{\step{k}}}\sum_{i=1}^nH_i-\E[H]\right|^2\right],
\end{align*}
so that
\[\E\left[|\hat{h}_{1,n^{\step{k}}}-\E[H]|^2\right]\le \frac{\tilde{C}_2}{n^{\step{k}}}\left(1+\frac{6}{\lambda_{\max}^2\epsilon^2}\right).\]
We proceed by analyzing the second term in the right hand side of \eqref{eq:trian}. Following the exact same procedure, noting that both $\smash{\hat{h}_{2,n^{\step{k}}}}$ and $\theta_2=\E[H^2]$ are majorized by $K$, again by \eqref{eq:BND},
\[ \E\left[|\hat{h}_{2,n^{\step{k}}}-\E[H^2]|^2\right]\le \frac{\tilde{C}_2}{n^{\step{k}}}\left(1+\frac{6K^2}{\epsilon^2}\right).\]

Then we consider the third term in the right-hand side of \eqref{eq:trian}. Now note that both $\smash{\hat{h}_{3,n^{\step{k}}}(\lambda)}$ and $\E[e^{-\lambda H}]$ are bounded above by 1. Hence,
\begin{align*}
\E\left[|\hat{h}_{3,n^{\step{k}}}(\lambda)-\E[e^{-\lambda H}]|^2\right]&\leq  \frac{\tilde{C}_2}{n^{\step{k}}}\left(1+\frac{6}{\epsilon^2}\right).
\end{align*}
For the remaining term, directly from \eqref{eq:sigma_H}, we obtain
\begin{align*}
\E\left[|\hat{h}_{4,n^{\step{k}}}(\lambda)-\E[-He^{-\lambda H}]|^2\right]&\leq \frac{\tilde{C}_2}{n^{\step{k}}}.
\end{align*}
We conclude that with 
\begin{align*}
C_2:={\tilde{C}_2}\left(4+\frac{6}{\epsilon^2}\left(\lambda_{\max}^{-2}+K^2+1\right)\right)
\end{align*}
the stated assertion \eqref{eq:theta_mixing} holds. 
\end{proof}

\begin{proof}[Proof of Lemma~\ref{lem:g_lip_bg}]
Since the service-time distribution is known, the functions
$\zeta(\lambda,\lambda_{\bg})$ and
$\zeta'(\lambda,\lambda_{\bg})$ are uniquely determined by
$\lambda_{\bg}$ through the fixed-point equation
\[
    \zeta(\lambda,\lambda_{\bg})
    =
    f_H^*(\lambda+\lambda_{\bg}-\lambda_{\bg}\zeta(\lambda,\lambda_{\bg}))
\]
and Equation \eqref{eq:zeta_lambda_prime}. Hence
$g^\circ(\lambda;\lambda_{\bg})$ is a continuously differentiable
function of the single parameter $\lambda_{\bg}$.

We first establish that all denominators appearing in
$g^\circ$ are uniformly bounded away from zero. Since
$\lambda\leq\lambda_{\max}$ and
$\lambda_{\bg}\in\Theta_\epsilon$,
\[
    \lambda+\lambda_{\bg}
    \leq
    \lambda_{\max}
    +
    \left(\frac{1}{\E[H]}-\epsilon-\lambda_{\max}\right)
    =
    \frac{1}{\E[H]}-\epsilon.
\]
Therefore,
\[
    1-(\lambda+\lambda_{\bg})\E[H]
    \geq
    \epsilon\E[H]>0.
\]
Moreover, since $\lambda\geq\lambda_{\min}>0$, the factors
$\lambda$ and $\lambda^2$ are uniformly bounded away from zero.
It thus remains to consider $\zeta(\lambda,\lambda_{\bg})$.  
With 
$a(\lambda,\lambda_{\bg})
:=
\lambda+\lambda_{\bg}-\lambda_{\bg}\zeta(\lambda,\lambda_{\bg})$, we have that $0\leq\zeta\leq1$ implies that
\[
    0\leq a(\lambda,\lambda_{\bg})
    \leq
    \lambda+\lambda_{\bg}
    \leq
    \frac{1}{\E[H]}-\epsilon .
\]
The Laplace transform $f_H^*$ is decreasing and strictly positive on
bounded intervals, and hence
\[
    \zeta(\lambda,\lambda_{\bg})
    =
    f_H^*(a(\lambda,\lambda_{\bg}))
    \geq
    f_H^*\left(\frac{1}{\E[H]}-\epsilon\right)
    =:c_\zeta>0 .
\]
Thus $1/\zeta$ and $1/\zeta^2$ are uniformly bounded.
Furthermore, by the stability condition, the denominator in Equation \eqref{eq:zeta_lambda_prime} is uniformly
bounded away from zero, while
\[
    |f_H'(x)|
    =
    \E[He^{-xH}]
    \leq
    \E[H].
\]
Hence $\zeta'$ is uniformly bounded on
$\mathcal L\times\Theta_\epsilon$.
Similarly, by differentiating the fixed-point equation and (29) with respect to
$\lambda_{\mathrm{bg}}$ and noting that
\[
    \left|f_H''(x)\right|
    = \E[H^2 e^{-xH}]
    \le \sup_{y \ge 0} y^2 e^{-\lambda_{\min}y}
    < \infty,
\]
we have that $\partial \zeta/\partial \lambda_{\mathrm{bg}}$
and $\partial \zeta'/\partial \lambda_{\mathrm{bg}}$
are also uniformly bounded on $\mathcal L\times\Theta_\epsilon$.
Consequently,
\[
    \sup_{\lambda\in\mathcal L}
    \sup_{\lambda_{\bg}\in\Theta_\epsilon}
    \left|
    \frac{\partial}{\partial\lambda_{\bg}}
    g^\circ(\lambda;\lambda_{\bg})
    \right|
    <\infty .
\]
Applying the mean value theorem yields a constant $C_1>0$ such that
\[
    \left|
    g^\circ(\lambda;\lambda_{\bg})
    -
    g^\circ(\lambda;\lambda_{\bg}')
    \right|
    \leq
    C_1|\lambda_{\bg}-\lambda_{\bg}'|,
\]
for all $\lambda\in\mathcal L$ and
$\lambda_{\bg},\lambda_{\bg}'\in\Theta_\epsilon$.
\end{proof}

    Prior to the proof of Lemma~\ref{lem:theta_BG} we establish an intermediate result on the mixing properties of the sequence $\{D_i\}_{i\in{\mathbb N}}$. This is a standard result but we did not find a direct reference and therefore verify the necessary steps here.
    \begin{lemma}\label{lem:D_mixing}
    Suppose $\E[H^2]<\infty$ and that $D_0=d_0\leq w$. 
    \begin{itemize}
    \item[(i)] For any measurable function $f:[0,\infty)\times\mathcal{L}\to[0,1]$ there exists a positive constant $K_1<\infty$ such that, for all $n\in{\mathbb N}$,
    \begin{equation}   \label{eq:D_mixing}     
    \sup_{\lambda\in\mathcal{L}}\sup_{d_0\leq w}
    \E_{\lambda}\left[
    \left|\frac{1}{n}\sum_{i=1}^n
    \left(f(D_i,\lambda)-\E_\lambda[f(D,\lambda)]\right)\right|^2
    \, \Big|\,  D_0=d_0
    \right]
    \leq \frac{K_1}{n} .
    \end{equation}
    \item[(ii)] Consider the parameterized class of functions $\mathcal{F}=\{f_\ell:[0,\infty)\times\mathcal{L}\to[0,1]\ : \ell\in[\ell_1,\ell_2]\}$. If for every $d\geq0$, $f_\ell(d,\lambda)$ is absolutely continuous (with respect to $\ell)$ and 
    \begin{align*} 
    \bar{m}:= \sup_{\ell\in[\ell_1,\ell_2],\ d\geq 0,\ \lambda\in\mathcal{L}}|\frac{\dd }{\dd \ell}f_\ell(d,\lambda)|<\infty , 
    \end{align*}
    then there exists a constant $K_2$ such that
    \begin{equation}   \label{eq:D_mixing_2}     
    \sup_{\lambda\in\mathcal{L}}\sup_{d_0\leq w}
    \E_\lambda\left[
    \sup_{\ell\in[\ell_1,\ell_2]}
    \left|
    \frac{1}{n}\sum_{i=1}^n
    \left(f_\ell(D_i,\lambda)-\E_\lambda[f_\ell(D,\lambda)]\right)
    \right|^2
    \, \Big|\,  D_0=d_0
    \right]
    \leq \frac{K_2}{n}. 
    \end{equation}
    \end{itemize}
    \end{lemma}
    \begin{proof}
    (i)~Let $W_i$ denote the waiting time in the queue of tagged packet $i$. Then
    $D_i=W_i+H_i$. The sequence $(H_i)_{i\in\mathbb N}$ is i.i.d., while
    $(W_i)_{i\in\mathbb N}$ is regenerative, with regeneration epochs given
    by arrivals that find the queue empty. Hence, defining $X_i:=f(D_i,\lambda)$,
    the sequence $(X_i)_{i\in\mathbb N}$ is also regenerative, with the same
    regeneration epochs. Let $N_j$ denote the length of regeneration cycle $j$,
    and define
    \begin{align*}
    C_1&:=\left\{1, 2, \ldots, N_1\right\},\\
    C_j&:=\left\{\sum_{k=1}^{j-1}N_k+1,
    \sum_{k=1}^{j-1}N_k+2, \ldots, 
    \sum_{k=1}^{j}N_k\right\},
    \qquad j=2,3,\ldots,\\
    U_j&:=\sum_{i\in C_j}\left(X_i-\E_\lambda[X]\right).
    \end{align*}
    If $D_0\leq w$, then the initial workload is also bounded, since
    $W_0\leq D_0\leq w$. The distribution of the initial cycle
    $(N_1,U_1,C_1)$ may therefore depend on $D_0$, whereas the subsequent
    cycles $(N_2,U_2,C_2), (N_3,U_3,C_3),\ldots$  are i.i.d.\ and independent of the
    initial state.

    Since we assume $0 \leq f(d,\lambda) \leq1$ ($d \geq 0$), we have $0\leq X_i\leq1$ and
    $0\leq\E_\lambda[X]\leq1$. It then follows that
    $|U_j|\leq N_j$ and $\E_\lambda[U_j^2]\leq\E_\lambda[N_j^2]$.
    Moreover, $\E_\lambda[U_j]=0$ for $j=2,3,\ldots$.

    For the $n$-th observation, write $T_n$ for the index of the regeneration
    cycle containing it. Letting
    \[
        R_n:=\sum_{\substack{i\in C_{T_n}\\i>n}}
        \left(X_i-\E_\lambda[X]\right)
    \]
    denote the centered reward remaining until the end of that cycle, we have
    \[
        \sum_{i=1}^n\left(X_i-\E_\lambda[X]\right)
        =\sum_{j=1}^{T_n}U_j-R_n,
        \qquad |R_n|\leq N_{T_n}.
    \]
    For $j\geq2$, the cycle lengths $N_j$ are i.i.d.; write $N$ for their
    common distribution. Since $T_n\leq n$,
    \[
        \sum_{j=2}^{T_n}U_j
        =\sum_{j=2}^nU_j\mathbf{1}_{\{T_n\geq j\}},
        \qquad
        \{T_n\geq j\}
        =\left\{\sum_{r=1}^{j-1}N_r<n\right\}.
    \]
    Hence, $\mathbf{1}_{\{T_n\geq j\}}$ is determined by the regeneration
    cycles preceding cycle $j$. Thus, for $2\leq j<k\leq n$,
    \begin{align*}
    \E_\lambda\!\left[
    U_jU_k\mathbf{1}_{\{T_n\geq j\}}
    \mathbf{1}_{\{T_n\geq k\}}
    \,\middle|\,D_0=d_0\right]
    &=
    \E_\lambda\!\left[
    U_jU_k\mathbf{1}_{\{T_n\geq k\}}
    \,\middle|\,D_0=d_0\right]=0,
    \end{align*}
    where the first equality uses
    $\{T_n\geq k\}\subseteq\{T_n\geq j\}$. For the second equality, note that
    $\{T_n\geq k\}$ is determined by the cycles preceding cycle $k$, whereas
    $U_k$ is independent of these cycles and satisfies
    $\E_\lambda[U_k]=0$. Consequently,
    \[
        \E_\lambda\!\left[
        \left(\sum_{j=2}^{T_n}U_j\right)^2
        \,\bigg|\,D_0=d_0\right]
        =
        \sum_{j=2}^n
        \E_\lambda\!\left[
        U_j^2\mathbf{1}_{\{T_n\geq j\}}
        \,\middle|\,D_0=d_0\right]
        \leq(n-1)\E_\lambda[N^2].
    \]

    Let
    \begin{align*}
    m_w:=
    \sup_{\lambda\in\mathcal{L}}\sup_{d_0\leq w}
    \E_\lambda\left[N_1^2\, \Big|\, D_0=d_0\right].
    \end{align*}
    By a coupling argument, the length of an ordinary regeneration cycle is
    stochastically dominated by that of the first cycle when $W_0=w$. Hence, as $W_0\leq D_0\leq w$,
    \[
        \sup_{\lambda\in\mathcal L}\E_\lambda[N^2]\leq m_w.
    \]
    Furthermore, since $|R_n|\leq N_{T_n}$ and $T_n\leq n$,
    \begin{align*}
    \E_\lambda\!\left[R_n^2\,\middle|\,D_0=d_0\right]
    &\leq
    \E_\lambda\!\left[N_{T_n}^2\,\middle|\,D_0=d_0\right]\leq
    \E_\lambda\!\left[N_1^2\,\middle|\,D_0=d_0\right]
    +(n-1)\E_\lambda[N^2].
    \end{align*}
    Combining the above bounds with
    $(a-b)^2\leq2a^2+2b^2$ yields, uniformly over $\lambda\in\mathcal L$ and
    $d_0\leq w$,
    \begin{align*}
    \E_\lambda\!\left[
    \left|
    \frac{1}{n}\sum_{i=1}^n
    \left(X_i-\E_\lambda[X]\right)
    \right|^2
    \,\middle|\,D_0=d_0\right]
    &\leq \frac{6m_w}{n}.
    \end{align*}
    Since $m_w<\infty$, \eqref{eq:D_mixing} follows with $K_1=6m_w$.

    \smallskip
    \noindent (ii)~Define
    \[
        Z_n(\ell):=\frac{1}{n}\sum_{i=1}^n f_\ell(D_i,\lambda)-\E_\lambda[f_\ell(D,\lambda)].
    \]
    Due to $Z_n(\ell)=Z_n(\ell_1)+\int_{\ell_1}^{\ell}Z_n'(u)\,\dd u$, Cauchy--Schwarz gives
    \[
        \sup_{\ell\in[\ell_1,\ell_2]}|Z_n(\ell)|^2
        \leq
        2|Z_n(\ell_1)|^2+
        2(\ell_2-\ell_1)\int_{\ell_1}^{\ell_2}|Z_n'(u)|^2\,\dd u.
    \]
    Since, by assumption, $\left|\frac{\dd}{\dd\ell}f_\ell(d,\lambda)\right|\leq\bar m$, dominated
    convergence gives
    \[
        Z_n'(u)
        =
        \frac{1}{n}\sum_{i=1}^n\frac{\dd}{\dd u}f_u(D_i,\lambda)
        -\E_\lambda\!\left[\frac{\dd}{\dd u}f_u(D,\lambda)\right].
    \]
    Applying part~(i) to
    $g_u(d):=(\frac{\dd}{\dd u}f_u(d,\lambda)+\bar m)/(2\bar m)\in[0,1]$
    therefore yields, uniformly in $u\in[\ell_1,\ell_2]$ and
    $\lambda\in\mathcal L$,
    \[
        \E_\lambda\!\left[|Z_n'(u)|^2\,\middle|\,D_0=d_0\right]
        \leq\frac{4\bar m^2K_1}{n}.
    \]
    Together with part~(i) applied to $f_{\ell_1}$, this gives
    \[
        \E_\lambda\!\left[
        \sup_{\ell\in[\ell_1,\ell_2]}|Z_n(\ell)|^2
        \,\middle|\,D_0=d_0\right]
        \leq
        \frac{2K_1}{n}
        +\frac{8(\ell_2-\ell_1)^2\bar m^2K_1}{n},
    \]
    which completes the proof.
    \end{proof}

\begin{proof}[Proof of Lemma~\ref{lem:theta_BG}]
\noindent \noindent It is observed that the estimator $\hat{\lambda}_{\bg}$ is a
     $Z$-estimator truncated at zero. Define $A_0:=\{\Psi_n(0)<0\}$. We use the decomposition
     \begin{align*}
     \E|\hat{\lambda}_{\bg}-\lambda_{\bg}|^2
     &=\E\!\left[|\hat{\lambda}_{\bg}-\lambda_{\bg}|^2\mathbf{1}_{A_0}\right]
     +\E\!\left[|\hat{\lambda}_{\bg}-\lambda_{\bg}|^2\mathbf{1}_{A_0^c}\right].
     \end{align*}
     On $A_0$, $\hat{\lambda}_{\bg}=0$. Thus the first term vanishes when
     $\lambda_{\bg}=0$, while for $\lambda_{\bg}>0$ it equals
     $\lambda_{\bg}^2\mathbb{P}(A_0)$, where
     \[
         \mathbb{P}(A_0)
         =
         \mathbb{P}\!\left(
         \frac{1}{n}\sum_{i=1}^n e^{-\lambda D_i}
         >1-\lambda\E[H]\right).
     \]
     Define $w(x):=1-(\lambda+x)\E[H]-\E[e^{-(\lambda+x)D}]$ on
     $[-\lambda,\lambda_{\bg}]$. Then $w(-\lambda)=w(\lambda_{\bg})=0$,
     $w''(x)=-\E[D^2e^{-(\lambda+x)D}]<0$, and
     $w(0)=1-\lambda\E[H]-\E[e^{-\lambda D}]>0$. Moreover, with
     $c_0:=\tfrac{1}{2}\,{\lambda_{\min}}\inf_{\lambda\in\mathcal L}
     \E[D^2e^{-(\lambda+\lambda_{\bg})D}]$, the concavity argument gives
     $w(0)\geq c_0\lambda_{\bg}$. Consequently,
     \[
         \mathbb{P}(A_0)
         \leq
         \mathbb{P}\!\left(
         \Bigg|
         \frac{1}{n}\sum_{i=1}^n e^{-\lambda D_i}
         -\E[e^{-\lambda D}]
         \Bigg|>c_0\lambda_{\bg}\right).
     \]
     Since $\E[H^2]<\infty$, part~(i) of Lemma~\ref{lem:D_mixing} and Chebyshev's inequality give
     \begin{equation}\label{eq:truncation_error_bound}
     \E\!\left[|\hat{\lambda}_{\bg}-\lambda_{\bg}|^2\mathbf{1}_{A_0}\right]
     \leq \frac{K_1}{c_0^2n}.
     \end{equation}

     On $A_0^c$, $\hat{\lambda}_{\bg}$ is an interior solution that satisfies
     $\Psi_n(\hat{\lambda}_{\bg})=0$. For $\ell\leq\lambda_{\bg}$, the
     concavity argument gives $|\lambda_{\bg}-\ell|
     \leq |w(\ell;\lambda,\lambda_{\bg})|/c_0$, while for
     $\ell\geq\lambda_{\bg}$, monotonicity of $w'$ gives the same bound with
     $c_1:=\inf_{\lambda\in\mathcal L}
     \{\E[H]-\E[De^{-(\lambda+\lambda_{\bg})D}]\}$. Thus, with
     $c:=\min\{c_0,c_1\}$,
     $|\lambda_{\bg}-\ell|\leq c^{-1}|w(\ell;\lambda,\lambda_{\bg})|$
     for $\ell\in\Theta_\epsilon$. Taking $\ell=\hat{\lambda}_{\bg}$ and
     using $\Psi_n(\hat{\lambda}_{\bg})=0$ yields
     \[
         |\hat{\lambda}_{\bg}-\lambda_{\bg}|
         \leq
         \frac{1}{c}
         \left|
         \frac{1}{n}\sum_{i=1}^n e^{-(\lambda+\hat{\lambda}_{\bg})D_i}
         -\E[e^{-(\lambda+\hat{\lambda}_{\bg})D}]
         \right|.
     \]
     Let $f_\ell(d,\lambda)=e^{-\ell d}\in[0,1]$. By part~(ii)  of Lemma~\ref{lem:D_mixing},
     \begin{align*}
     \E\!\left[
     |\hat{\lambda}_{\bg}-\lambda_{\bg}|^2\mathbf{1}_{A_0^c}
     \right]
     &\leq
     \frac{1}{c^2}
     \E\!\left[
     \sup_{\ell\in\Theta_\epsilon}
     \left|
     \frac{1}{n}\sum_{i=1}^n e^{-(\lambda+\ell)D_i}
     -\E[e^{-(\lambda+\ell)D}]
     \right|^2
     \right]\leq \frac{K_2}{c^2n}.
     \end{align*}
     Combining \eqref{eq:truncation_error_bound} and the preceding bound gives
     \[
         \E|\hat{\lambda}_{\bg}-\lambda_{\bg}|^2
         \leq
         \left(
         \frac{K_1}{c_0^2}
         +\frac{K_2}{\min\{c_0,c_1\}^2}
         \right)\frac{1}{n},
     \]
     which implies the stated claim.
\end{proof}

\section{Delta-method interval estimates for Example 1}
\label{app:delta}

This appendix derives the confidence intervals used in Subsection~\ref{sec:sim_A}. We omit the iteration index $k$ throughout; $\lambda$ is the fixed rate and $n$ denotes the number of informative packets observed. Because the service times are recovered exactly from the observed timestamps by inverting the Lindley recursion, the sample available at a decision epoch is i.i.d., and $\hat{\bm{\theta}}^{\step{k}}$ is a vector of sample means of
\begin{equation*}
{\boldsymbol Z}_i(\lambda) := \bigl(H_i,\; H_i^2,\; e^{-\lambda H_i},\; -H_i e^{-\lambda H_i}\bigr)^{\!\top}.
\label{eq:Z-A}
\end{equation*}
Provided $\E[H^4]<\infty$, the multivariate central limit theorem yields
\begin{equation*}
\sqrt{n}\,\bigl(\hat{\bm{\theta}}-{\bm{\theta}}_\lambda\bigr) \darrow \mathrm{N}\bigl(\mathbf{0},\,\Sigma(\lambda)\bigr),
\qquad
\Sigma_{jk}(\lambda) := \mathrm{Cov}\bigl(Z_{j,1}(\lambda),Z_{k,1}(\lambda)\bigr).
\label{eq:CLT-theta-A}
\end{equation*}
By the argument in the proof of Lemma~\ref{lem:theta_FD},
$\mathbb P(\hat{\bm\theta}\notin\Theta_\epsilon)\to0$, so the truncation is
asymptotically inactive.
Since ${\bm\theta}\mapsto g(\lambda;{\bm\theta})$ is smooth on $\Theta_\epsilon$, the delta method yields
\begin{equation*}
\sqrt{n}\,\bigl(g(\lambda;\hat{\bm{\theta}})-g(\lambda;{\bm{\theta}}_\lambda)\bigr) \darrow \mathrm{N}\bigl(0,\,\sigma_g^2(\lambda)\bigr),
\quad
\sigma_g^2(\lambda) := J(\lambda)^{\!\top}\Sigma(\lambda)J(\lambda),
\label{eq:CLT-g-A}
\end{equation*}
with $J\equiv J(\lambda):=\nabla_{\bm\theta}\,g(\lambda;{\bm\theta}_\lambda)$. Let ${\bm\theta}_\lambda:=(h_1,h_2,h_3,h_4)$, then differentiating \eqref{eq:psi_grad-0} yields
\begin{equation*}
\begin{aligned}
J_1 &= \frac{\lambda h_2}{(1-\lambda h_1)^{3}}+\frac{h_4}{h_3^{2}},
&\quad
J_2 &= \frac{1}{2(1-\lambda h_1)^{2}},
\\[2pt]
J_3 &= \frac{2h_4}{\lambda h_3^{3}}+\frac{1}{\lambda^2 h_3^{2}}-\frac{2h_1h_4}{h_3^{3}},
&\quad
J_4 &= \frac{\lambda h_1-1}{\lambda h_3^{2}}.
\end{aligned}
\label{eq:J-A}
\end{equation*}
Replacing $\Sigma$ by the sample covariance of the ${\boldsymbol Z}_i(\lambda)$ and $J$ by $\hat J = \nabla_{\bm\theta}\,g(\lambda;\hat{\bm\theta})$ yields a consistent estimator $\hat\sigma_g^2(\lambda)$ and the asymptotic nominal $95\%$ interval
\begin{equation}
g\bigl(\lambda;\hat{\bm\theta}\bigr) \;\pm\; z_{0.975}\,\frac{\hat\sigma_g(\lambda)}{\sqrt{n}} .
\label{eq:CI-A}
\end{equation}
The same construction applied to $\E_\lambda[A]$ in \eqref{eq:EA} gives an interval for the objective, with Jacobian
\begin{equation*}
J^{A} = \Biggl(\frac{\lambda^2 h_2}{2(1-\lambda h_1)^{2}}+1-\frac{1}{h_3},\;
               \frac{\lambda}{2(1-\lambda h_1)},\;
               -\frac{1-\lambda h_1}{\lambda h_3^{2}},\; 0\Biggr)^{\!\top},
\label{eq:JA-A}
\end{equation*}
whose last coordinate vanishes because $\E_\lambda[A]$ does not involve $f_H'(\lambda)$.

\end{document}